\documentclass[sigconf,authorversion,nonacm,balance=true]{acmart}

\usepackage[switch]{lineno} 
\AtBeginDocument{%
  }

\usepackage{todonotes}

\usepackage{graphicx} % Required for inserting images
\usepackage{etoolbox}
\usepackage{algorithm}
\usepackage{algpseudocode}
\usepackage{amsmath}
\usepackage{mathtools}
\usepackage{amsthm}
\usepackage{placeins}
\usepackage{color}
\usepackage{parskip}
\usepackage{enumitem}
\usepackage{tikz}
\usepackage{tikz-uml}
\usetikzlibrary{arrows,shadows,calc}
\usepackage{pgf-umlsd}
\usepackage{pst-node}
\usepackage{tikz-cd}
\usepackage[noabbrev,capitalize]{cleveref}

\floatname{algorithm}{Experiment} % Change name from Algorithm 
\Crefname{algorithm}{Experiment}{Experiments}

\usepackage{float} % Required to define custom floating environments
\newfloat{oracle}{tbp}{loe}
\floatname{oracle}{Oracle}
\crefformat{oracle}{Oracle~#2#1#3}
\Crefname{oracle}{Oracle}{Oracles}

\usepackage{eso-pic}

\AddToShipoutPictureFG*{
  \AtPageUpperLeft{
    \hspace*{1cm}
    \raisebox{-1cm}{
      \parbox{0.5\textwidth}{
        \noindent\small\textsf{
          This paper was published in the Proceedings of the ACM Asia Conference on Computer and Communications Security (ASIA CCS '26), June 1--5, 2026, Bangalore, India.
          \url{https://doi.org/10.1145/3779208.3785289}
        }
      }
    }
  }
}

\newfloat{construction}{tbp}{loe}
\floatname{construction}{Construction}
\crefformat{construction}{Construction~#2#1#3}
\Crefname{construction}{Construction}{Constructions}

\newcommand\NoThen{\renewcommand\algorithmicthen{}}
\newtheorem{theorem}{Theorem}[section]
\theoremstyle{definition}
\newtheorem{definition}{Definition}[section]

\newcommand{\olive}[1]{#1}
\newcommand{\Oracle}{{\mathcal{O}}}
\newcommand{\adv}{{\mathcal{A}}}
\newcommand{\mpk}{{mpk}} %master public key
\newcommand{\msk}{{msk}} %master private key
\newcommand{\upk}{{upk}} %user public key
\newcommand{\usk}{{usk}} %user secret key
\newcommand{\sppk}{{sppk}} %SP public key
\newcommand{\spsk}{{spsk}} %SP secret key

\newcommand{\COMMENT}[1]{\textcolor{gray}{\textit{// #1}}}

\newcommand{\Sig}{{\Sigma}}
\newcommand{\Enc}{{\Lambda}}
\newcommand{\NIKE}{{\Psi}}
\newcommand{\NIZK}{{\Pi}}

\newcommand{\ZKP}{\mathsf{NIZK}}

\newcommand*\circled[1]{\tikz[baseline=(char.base)]{
            \node[shape=circle,draw,inner sep=1.2pt] (char) {{\footnotesize #1}};}}

\newcommand{\nym}{\mathsf{nym}}
\newcommand{\pk}{\mathsf{pk}}
\newcommand{\sk}{\mathsf{sk}}
\newcommand{\msg}{\mathsf{msg}}

\newcommand{\shk}{\mathsf{shk}}

\newcommand{\getsr}{\gets_{\$}}

\newcommand{\GG}{\mathbb{G}}
\newcommand{\ZZ}{\mathbb{Z}}

\newcommand{\secpar}{\lambda}
\newcommand{\pp}{pp}

\newcommand{\alg}[1]{\mathsf{#1}}

\newcommand{\NymGen}{\alg{NymGen}}

\newcommand{\Groth}{\alg{Groth}}
\newcommand{\ElGamal}{\alg{ElGamal}}
\newcommand{\DHN}{\alg{DH}}
\newcommand{\ParGen}{\alg{Setup}}
\newcommand{\KeyGen}{\alg{KeyGen}}
\newcommand{\Sign}{\alg{Sign}}
\newcommand{\Rand}{\alg{Rand}}
\newcommand{\NymVerify}{\alg{NymVf}}
\newcommand{\Verify}{\alg{Verify}}
\newcommand{\ShareKey}{\alg{ShareKey}}
\newcommand{\Encrypt}{\alg{Enc}}
\newcommand{\Decrypt}{\alg{Dec}}
\newcommand{\Open}{\alg{Open}}
\newcommand{\Prove}{\alg{Prove}}

\newcommand{\bpk}{bPk}
\newcommand{\bpkplus}{bPk\textsuperscript{\#}}

\renewcommand{\longleftarrow}{\gets}

\begin{document}
\title{\bpkplus: Delegatable Pseudonyms}
\subtitle{And Their Applications to National eID Systems}

    \author{Stephan Krenn}
      \affiliation{%
     \institution{AIT Austrian Institute of Technology}
     \city{Vienna}
     \country{Austria}}
    \email{stephan.krenn@ait.ac.at}
    
    \author{Doryan Lesaignoux}
    \affiliation{%
     \institution{AIT Austrian Institute of Technology}
     \city{Vienna}
     \country{Austria}}
    \email{doryan.lesaignoux@ait.ac.at}
    
    \author{Sebastian Ramacher}
    \affiliation{%
     \institution{AIT Austrian Institute of Technology}
     \city{Vienna}
     \country{Austria}}
    \email{sebastian.ramacher@ait.ac.at}


%\renewcommand{\shortauthors}{Krenn et al.}

\begin{abstract}
  Electronic identities (eIDs) are crucial in an increasingly digitalized environment. 
  Pseudonyms, as offered by Austria's governmental sector-specific personal identifiers (bPks), can significantly improve privacy by ensuring that personal data is not universally traceable across public services and private companies. 
  However, the current architecture comes with several challenges regarding availability, privacy, and authenticity, due to a fully centralized design. 

  This paper proposes \bpkplus, a distributed architecture to address these issues, reducing reliance on the central authority, while still providing all functional requirements to the existing bPk system.
  In particular, users are delegated the rights to compute their own pseudonyms, thereby minimizing metadata revealed to the central authority, while (subsets of) service providers may receive the right to compute pseudonyms only within their own domain, thereby reducing the availability needs of the central authority.

  To the best of our knowledge, we provide the first formal framework for such \emph{delegatable pseudonym systems}, together with a generic construction for which we provide formal security proofs.
  Furthermore, we propose a concrete instantiation of our construction, together with a reference implementation demonstrating the practical efficiency.
\end{abstract}


\maketitle
\nolinenumbers

\section{Introduction}\label{s:intro}
Electronic identification schemes (eIDs) are essential for secure, efficient, and reliable digital interactions in both the public and private sectors. 
They enable individuals to authenticate themselves online, e.g., when accessing public e-government services, or private services like banking or healthcare, without relying on insecure methods like passwords. 
By providing a verified digital identity, eIDs help reduce fraud, identity theft, and cybercrime while ensuring compliance with data protection laws like the GDPR.

In this context, pseudonyms play a crucial role in enhancing privacy and data protection while enabling secure digital interactions. 
Instead of using a fixed personal identifier across all services, pseudonyms allow users to generate sector-specific or transaction-specific identifiers, ensuring that their real identity remains protected in different contexts. 
This approach aligns with the principles of GDPR, particularly data minimization and purpose limitation, by allowing only the necessary identity attributes to be shared on a need-to-know basis, and also helps prevent unwanted tracking and profiling, reducing the risk of mass surveillance or identity theft. 

While many European countries have eID systems, not all of them implement sector-specific pseudonyms. 
Instead, national eID systems often use single identifier (e.g., the social security numbers) without the added complexity of domain-specific pseudonyms.

One of the first countries to introduce pseudonyms was Austria, where so-called  ``bereichsbezogene Personenkennzeichen'' (bPk; sector-specific personal identifiers)\footnote{\url{https://www.bundeskanzleramt.gv.at/agenda/digitalisierung/stammzahlenregisterbehoerde/bereichsspezifische-personenkennzeichen/beschreibung.html}} to enhance data privacy, security, and interoperability in digital identity management were already introduced in 2004 as part of the national e-Government Act.\footnote{\url{https://ris.bka.gv.at/GeltendeFassung.wxe?Abfrage=Bundesnormen&Gesetzesnummer=20003230}} 
These identifiers ensure that personal data is not universally traceable across different administrative or private sectors, reducing the risk of unauthorized profiling and misuse. 
By deriving unique identifiers per sector, ID Austria ensures that individuals can interact with government services, healthcare providers, and financial institutions without exposing a single, universal personal identifier.

In the following we briefly recap the functioning of the information flow on a logical level, cf. also \cite{DBLP:conf/birthday/PoschPTZZ11} and \cref{fig:bpks_old}. 
For more details on these flows, we also refer to \cref{sec:sequence-diagrams}.
The actual authentication system starts with a registration phase. All Austrian citizens and new residents are automatically created by the central 
authority; \olive{this registration and participation are mandatory by law from first official registration (e.g., after birth or immigration) in Austria onward.}
When a user wishes to pseudonymously authenticate to a public service $SP$ \circled{1}, the user leverages the national eID system \circled{2} to authenticate to a dedicated central public authority (“Stammzahlregisterbehörder”) \circled{3} upon request of the service provider. This authority derives the bPk from a locally stored and protected secret key created during the registration and sends it to the service provider \circled{4}. Furthermore, also in the absence of users, public authorities and governmental services may request the central authority to derive a bPk for a specific user, thereby relieving them from locally storing identity data in the plain to adhere to privacy-by-design principles.
Additionally, the central authority is the only entity in the system which is able to link bPks across different domains, in case that multiple public authorities are involved in a transaction (e.g., municipality services and tax authorities), e.g., $SP'$ can obtain an encrypted version of the bPk for $SP''$ to enable $SP''$ to assign the bPk to the current file \circled{5}.
In case of abuse of pseudonymity, e.g., by filing false statements, or if necessary for other reasons (e.g., decease of a user), the central authority is also able to de-anonymize bPks.
Finally, starting with $2018$, bPks cannot just be used in communications with public authorities, but also with private companies. 
Overall, this caused more than $2.5$ billion queries to the central authority annually~\cite{szrb}.

\begin{figure}[bt!]
  \includegraphics[width=0.76\columnwidth]{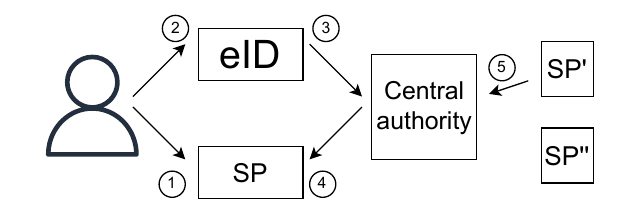}
  \caption{Logical information flow in the current bPk architecture}
  \label{fig:bpks_old}
\end{figure}

Cryptographically, the computation of the bPk follows an encrypt-then-hash approach.
That is, the central authority first encrypts the user's resident registration number $uid$ using a symmetric encryption scheme $\mathsf{SymEnc}$ under a secret key $sk$, and then hashes the result together with a public identifier of $SP$, obtaining the bPk.
That is, the bPk has the form $\mathsf{Hash}(\mathsf{SymEnc}(sk,uid),SP)$~\cite{id-austria-whitepaper}. 
Specifically, this approach is instantiated using \textsf{3DES} as symmetric encryption scheme, and \textsf{SHA-1} as hash function\footnote{\url{https://www.bundeskanzleramt.gv.at/agenda/digitalisierung/stammzahlenregisterbehoerde/veroeffentlichungen.html}}, while the supported domains are defined in a dedicated regulation\footnote{\url{https://www.ris.bka.gv.at/GeltendeFassung.wxe?Abfrage=Bundesnormen&Gesetzesnummer=20003476}}.

Looking at the existing solution, there are multiple questions to be addressed:

\begin{description}
    \item[Availability.]
      The current architecture of the Austrian pseudonym system relies on a fully centralized architecture.
      Even if protected with state-of-the-art techniques, a denial of service attack, or also a non-malicious outage of the central authority cannot be excluded, resulting in immediate impacts on public authorities which no longer can obtain bPks.
      Furthermore, with the increasing use in the private sector, this impact also increasingly effects private service providers such as banks or private healthcare providers.
    \item[Privacy.]
      In the current architecture, the central authority is a fully trusted entity, without the user having the freedom to choose between different providers.
      While this is somewhat inherent due to its defined tasks and duties, it is also involved in any authentication session involving bPks, i.e., it learns who is authenticating to which service when and how often.
      While the authority acts within a tight legal corset to guarantee that the collected data is not misused, a malicious adversary obtaining (potentially read-only) access to this data can infer detailed identity profiles of millions of citizens. 
    \item[Authenticity.]
      Overcoming the aforementioned challenges by simply distributing, e.g., $\mathsf{SymcEnc}(sk,uid)$ to the individual users does not work in a straightforward manner, \olive{as was indeed done by the Austrian ``Bürgerkarte'' (citizen's card)~\cite{theuermann19}}.
      However, while $sk$ remains protected at the central authority in this case and citizens may derive arbitrary bPks from this value, they could not prove to a service provider that the bPk is well-formed, i.e., a service provider could not distinguish a valid bPk for a given user from, e.g., a random bit string.
      \olive{This in turn allows users to authenticate under arbitrary bPks anytime they authenticate, which is especially critical for authentications to private companies, which do not have a legal basis to compare the correctness of the computation with the central authority.}
      Thus, such authenticity guarantees are essential due to the potential legal implications bound to user authentication.
    \item[Provability.]
      Finally, while the specific way how bPks are currently computed intuitively seems to achieve all expected security guarantees, at least when modeling the hash function as a random oracle, it does not provide formal guarantees.
      Specifically, neither a rigorous definition of the required properties nor formal proofs have been published, in contradiction to the provable security paradigm for highly sensitive applications. 
\end{description}

\paragraph{Our Contribution.}
The goal of this work is thus to suggest \bpkplus, a revised architecture and design, overcoming the above limitations without hindering the applicability of the proposed solution.
It is worth noting that while our approach is inspired by a specific real-world system, the solution does not contain specificities, e.g., of the related eID system, such that it could be used also in the context of other countries.

More precisely, our contributions can be summarized as follows:
\begin{itemize}
    \item 
      In a first step, we derive the functional and non-functional requirements from the current design of the Austrian bPk system, including legal constraints from the Austrian e-Government law (e.g., related to mandatory key generation at the time of birth).
      We summarize the different capabilities that have to be offered by such a system, and for the first time formally define the non-functional, i.e., security- and privacy-related, requirements. 
    \item 
      We then provide a provably secure generic construction satisfying the security definitions.
      In contrast to the existing solution, our solution is fully distributed in the sense that the central authority is only involved in unavoidable cases, e.g., when pseudonyms need to be linked across different domains, or when a pseudonym has to be deanonymized.
      Nevertheless, despite being computed by the user, the service provider has formal authenticity guarantees that a real user has provided a well-formed pseudonym.
    \item 
      In order to further reduce the dependability on the central authority, we further offer the possibility that also pre-defined service providers can compute pseudonyms for arbitrary users -- yet only for their own scope.
      While this solves the impact, e.g., of an attack on the central authority, this might not be desirable for certain (e.g., private) service providers.
      Our construction thus naturally allows for a distinction between service providers who may, and others who may not, compute pseudonyms.
    \item 
      Finally, we provide a concrete instantiation of our construction from standard building blocks such as digital signatures, zero-knowledge proofs, and non-interactive key exchange protocols. Going a step further, we instantiate our system with pairing-based building blocks and analyze its efficiency. Notably, on midrange hardware from 2022 the generation of a pseudonym including the proof of well-formedness takes less than $5$ ms on average whereas verification of the proofs takes less than $10$ ms.
\end{itemize}

\paragraph{Related Work.} 
\olive{There is a large body of work on privacy-preservation in eID systems. In the following we give a concise overview over some of the most closely related approaches found in the literature.}

The evolution of the Austrian eID system including its integration into eGovernment services has seen discussion in the academic literature throughout the last two decades, e.g., \cite{DBLP:journals/clsr/Rossler08,DBLP:conf/egov/TauberR09,DBLP:conf/birthday/PoschPTZZ11,DBLP:conf/sec/ZwattendorferS13,DBLP:journals/compsec/ZwattendorferS15,DBLP:journals/istr/ZwattendorferS16}. To the best of our knowledge, the specific requirements, design and the construction of bPks has not seen a formal treatment with an analysis in the framework for provable security.

For the German eID system~\cite{dagdelen2013cryptographic}, Dagdelen proposed a domain-specific pseudonymous signature scheme. In the German system, however, users are equipped with identity cards that store their secret key material. The generation of pseudonyms is done by running a non-interactive Diffie-Hellmann key exchange with a service provider public-key. In contrast, no proof is provided and thereby is affected by the same issues as the existing Austrian approach.

More generally, cryptographic pseudonym systems were introduced by Chaum~\cite{DBLP:journals/cacm/Chaum85}, and later formalized and deployed in a series of work, including scope-exclusive pseudonyms, e.g., \cite{SAC:LRSW99,SAC:CKLMNP15,CCS:CamVan02,PaquinZaverucha2013UPCS}. However, these concepts are focusing on user-centricity without the possibility for fine-grained delegation of the pseudonym computation to other entities.

\olive{Notably, Bringer et al.~\cite{DBLP:conf/fc/BringerCLP14} use domain-specific pseudonyms in an eID context.
In another series of work, Camenisch et al. \cite{camenisch2015linkable,camenisch2017privacy}  introduced domain-specific pseudonym systems which allow a central authority to translate pseudonyms from one domain to another providing very high anonymity and transparency.
While all being suitable candidates upon first glance, these schemes use models where keys are managed in a fully decentralized manner, thus being incompatible with the requirements of the Austrian legal framework.}

\olive{Deswarte et al.~\cite{deswarte2010proposal}  proposed a privacy-preserving national identity card system. However, their system was designed to prove attributes of a user in an unlinkable manner, and explicitly does not support pseudonyms. Furthermore, no rigorous security framework nor formal proofs are provided.}

Finally, in general Self-Sovereign Identity (SSI) systems such as \cite{DBLP:conf/trustcom/AbrahamKMRS21,SP:MMZJFK21}, users are completely put in control of their secret keys and thus all actions in the system involve active participating of users. As discussed above, bPks need to be computable in specific cases due to the legal requirements of the Austrian eID system and thus SSI systems lack required functionality.

\olive{In summary, while a large body of relevant research exists, to the best of our knowledge none of these works covers all the details identified in our analysis of the Austrian bPk system, which inherently requires a central authority, yet also requires the need to delegate pseudonym computation to both users and service providers in a provably secure manner, to increase privacy, overcome efficiency challenges, and maximize availability and resilience.}

\paragraph{Outline.}
  This document is structured as follows.
  \Cref{s:preliminaries} introduces the notation, and summarizes the necessary cryptographic background.
  In \cref{s:framework} we then introduce the modeling, including syntax and interfaces, but also functional and security requirements.
  We then present a generic construction from existing cryptographic building blocks, together with a rigorous security proof, in \cref{s:generic}, before presenting a concrete instantiation and implementation in \cref{s:instantiation}.
  In \cref{s:discussions}, we discuss the integration of revocation and key rotation in our construction.
  We finally conclude in \cref{s:conclusion}.

\section{Preliminaries}\label{s:preliminaries}
Readers familiar with the high-level concepts of public key encryption, digital signatures, non-interactive key exchange, and zero-knowledge proofs can safely skip this section upon first reading, and later come back for formal details, e.g., related to security proofs.

\subsection{Notations}

In this paper, we denote the security parameter as $\lambda \in \mathbb{N}$ and $1^\lambda$ as its unary representation. 

An algorithm $\adv$ is called \textit{Probabilistic Polynomial Time} (PPT) if its running time is bounded by a polynomial $P(|x|)$ for every input $x$. 
Unless stated otherwise, all algorithms and adversaries considered in this paper are PPT.

%The advantage of an adversary $\adv$ is denoted as $Adv_\adv(\lambda)$.

Let $(\Omega,\mathcal{E}, \mathbb{P})$ be a probability space. 
We write $\mathbb{P}[\Omega : \omega ]$ the probability of an event $\omega \in \mathcal{E}$ in space $\Omega$. A function $\varepsilon : \mathbb{N} \rightarrow \mathbb{R}^+$ is said negligible if it vanishes faster than every reverse polynomial. 
More formally, if $\forall k\in \mathbb{N}, \exists n_0 \in \mathbb{N} : \forall n > n_0, \varepsilon(n) \leq n^{-k}$.

\olive{For all cryptographic primitives used in this paper, we assume that the $\pk$ can always be implicitly derived from the corresponding $\sk$, i.e., knowing $\sk$ enables at least implicit access to $\pk$.}

\subsection{Cryptographic Building Blocks}
In the following recap the basic definitions required for the building blocks used in our generic construction.
For all the building blocks, we omit formal definitions of correctness, i.e., all schemes are supposed to function properly if all parties follow the protocol specifications.
\subsubsection{Public Key Encryption}
A public key encryption scheme is a cryptographic system that uses a publicly known key for encryption and a corresponding private key for decryption, enabling secure communication without prior key exchange.
\begin{definition}[\textbf{Public-Key Encryption scheme}]
    A public-key encryption scheme $\Enc$ is a tuple of four PPT algorithms $(\ParGen, \KeyGen, \allowbreak \Encrypt, \Decrypt)$ such that:
    \begin{itemize}
        \item $\ParGen(1^\lambda) \rightarrow \pp$: This algorithm takes a security parameter $\lambda$ and outputs public parameters $\pp$.
        \item $\KeyGen(\pp) \rightarrow (\sk,\pk)$: This algorithm takes public parameters $\pp$ and outputs a secret and public key $(\sk,\pk)$.
        \item $\Encrypt(\pk,m) \rightarrow c$: This algorithm takes a public key $\pk$ and a message $m$ and outputs a ciphertext $c$.
        \item $\Decrypt(\sk,c) \rightarrow m^*$: This algorithm takes a secret key $\sk$ and a ciphertext $c$ and outputs a message $m^*$.
    \end{itemize}
\end{definition}

The most fundamental security requirement for public key encryption requires that, only knowing the public key, it is computationally infeasible to decide which of two (adversarially chosen) plaintexts is encrypted in a ciphertext.
For an in-depth discussion, we refer, e.g., to Katz and Lindell~\cite{DBLP:books/crc/KatzLindell2014}.

\begin{definition}[\textbf{IND-CPA}]
    Let $\Enc$ be a public-key encryption scheme. We define the IND-CPA  experiment of  $\Enc$ by $\mathrm{Exp}_{\Enc,\adv}^{\mathrm{IND-CPA}}(\lambda)$ in \cref{ind-cpa}.     
    We say that $\Enc$ is IND-CPA secure if there is a negligible function such that: 
    \[\mathrm{Adv}_{\Enc,\adv}^{\mathrm{IND-CPA}}(\lambda) \overset{\mathrm{def}}{=} \bigg|\:\mathbb{P}\left[ \mathrm{Exp}_{\Enc,\adv}^{\mathrm{IND-CPA}}(\lambda) = 1 \right] - \frac{1}{2} \bigg| \leq \varepsilon(\lambda)\] 

%    \bigg|\:\mathbb{P}\left[
%    \mathrm{Exp}_{\Enc,\adv}^{\mathrm{IND-CPA}}(\lambda) = 1 \right] - \frac{1}{2}\bigg|

\begin{algorithm}[ht!]
\caption{$-$ $\mathrm{\mathbf{Exp}}_{\Enc,\adv}^{\mathrm{IND-CPA}}(\lambda)$}\label{ind-cpa}
\begin{algorithmic}[1]
\State $\pp \longleftarrow \ParGen(1^\lambda)$
\State $(\sk,\pk) \longleftarrow \KeyGen(\pp)$
\State $m_0,m_1 \longleftarrow 
\adv(\pp,\pk) \hspace{32mm} m_0,m_1 \in \mathcal{M}$ 
\State $b \overset{\$}{\longleftarrow}\{0,1\}$
\State $c^* = \Encrypt(\pk,m_b)$
\State $b^* \longleftarrow \adv(\pp,\pk,c^*)$
\If{$b^* = b $}
    \Return $1$
\EndIf
\State \Return $0$
\end{algorithmic}
\end{algorithm}

\begin{comment}

\begin{figure}[ht!]
    \centering
    \fbox{
\begin{minipage}{\columnwidth}
$\mathrm{\mathbf{Exp}}_{\Enc,\adv}^{\mathrm{IND-CPA}}(\lambda)$:
\begin{description}
  \item[\hspace{10mm}] $\pp \longleftarrow \ParGen(1^\lambda)$
  \item[\hspace{10mm}] $(\sk,\pk) \longleftarrow \KeyGen(\pp)$
  \item[\hspace{10mm}] $m_0,m_1 \longleftarrow 
\adv(\pp,\pk) \qquad \qquad ,m_0,m_1 \in \mathcal{M}$
  \item[\hspace{10mm}] $b \overset{\$}{\longleftarrow}\{0,1\}$
  \item[\hspace{10mm}] $c^* = \Encrypt(m_b,\pk)$
  \item[\hspace{10mm}] $b^* \longleftarrow \adv(\pp,\pk,c^*)$
  \item[\hspace{10mm}] If $b^* = b :$ Return 1
  \item[\hspace{10mm}] Return 0

\end{description}
\end{minipage}
}
    \caption{IND-CPA game}
    \label{ind-cpa}
\end{figure}
\end{comment}

\end{definition}

\subsubsection{Digital signature}
A digital signature ensures the authenticity, integrity, and non-repudiation of a message or document using a private signing key and a publicly verifiable signature.
\begin{definition}[\textbf{Signature scheme}]
    Under a security parameter $\lambda$, a signature scheme $\Sig$ is a tuple of four PPT algorithms $\ParGen,\KeyGen, \Sign, \Verify\}$ such that:
    \begin{itemize}
        \item $\ParGen(1^\lambda) \rightarrow \pp$: This algorithm takes a security parameter $\lambda$ and outputs public parameters $\pp$. 
        \item $\KeyGen(\pp) \rightarrow (\sk,\pk)$: This algorithm takes public parameters $\pp$ and outputs a secret and public key $(\sk,\pk)$.
        \item $\Sign(m,\sk) \rightarrow \sigma$: This algorithm takes a message $m$ and a secret key $\sk$ and outputs a signature $\sigma$.
        \item $\Verify(\pk,\sigma,m) \rightarrow 0/1$: This algorithm takes a public key $\pk$, a signature $\sigma$ and a message $m$ and outputs 1 or 0 whether the signature is valid or not.
    \end{itemize}   
\end{definition}

There exist several unforgeability notions. 
For our construction, we require only a basic notion, which states that no adversary not knowing the secret key can create a valid signature on an arbitrary new message, even after having seen arbitrarily many signatures on messages of its choice.
For an in-depth discussion, we refer, e.g., to Katz and Lindell~\cite{DBLP:books/crc/KatzLindell2014}.
\begin{definition}[\textbf{EUF-CMA}]
    Let $\Sig$ be a signature scheme. $\Sig$ achieves existential unforgeability under chosen-message attacks (EUF-CMA) if: 
    \[\mathbb{P}\left[\begin{array}{c}
    \pp\leftarrow\ParGen(1^\secpar)\\
    (\sk,\pk) \leftarrow \KeyGen(\cdot) \\
    (m^*,\sigma^*) \leftarrow \adv^{\Oracle^{\Sign}}(\pk)
    \end{array} :  
    \begin{array}{c} 
     \Verify(\pk,\sigma^*,m^*)) = 1 \\
     m^* \notin \mathcal{Q}^{\Sign}
    \end{array}
    \right] \leq \varepsilon(\lambda)\]
\end{definition}

\subsubsection{Non-Interactive Zero-Knowledge Proof of Knowledge Systems}
A zero-knowledge proof of knowledge is a cryptographic proof that allows a prover to convince a verifier that a statement is true without revealing any additional information beyond what is revealed by the message itself. 
If it only consists of a single message being sent from the prover to the verifier, it is called non-interactive (NIZK).

Formally, let $\mathcal{L}$ be a NP-language associated with a relation $R$ such that $\mathcal{L}_R = \{x\ | \exists w : R(x,w) = 1\}$

\begin{definition}[\textbf{Non Interactive proof}]
    A non-interactive proof system $\NIZK$ is a tuple of three PPT algorithms $\{\ParGen, \Prove, \Verify\}$ such that:
    \begin{itemize}
        \item $\ParGen(1^\lambda) \rightarrow crs$: This algorithm takes a security parameter $\lambda$ and outputs a common reference string $crs$.
        \item $\Prove(crs,x,w) \rightarrow \pi$: This algorithm takes a common reference string $crs$, a statement $x$ and a witness $w$ as input and outputs a proof $\pi$.
        \item $\Verify(crs,x,\pi) \rightarrow 0/1$: This algorithm takes a common reference string $crs$, a statement $x$ and a proof $\pi$ as input and outputs 1 or 0 whether the proof is valid or not.
    \end{itemize}
\end{definition}

Firstly, it must not be possible to generate proofs for wrong statements.
\begin{definition}[\textbf{Soundness}]
    $\NIZK$ is sound if for any PPT adversary $\adv$, there is a negligible function $\varepsilon$ such that: 
    \[\mathbb{P}\left[\begin{array}{c}
    \pp \leftarrow \ParGen(1^\lambda) \\
    (x^*,\pi^*) \leftarrow \adv(\pp)
    \end{array} \; :\;  
    \begin{array}{c} 
     \Verify(crs,x^*,\pi^*) = 1 \\
     \land\ \  x^* \notin \mathcal{L}_R
    \end{array}
    \right] \leq \varepsilon(\lambda)\]   

\end{definition}

Secondly, the verifier must not learn anything beyond the statement itself.
\begin{definition}[\textbf{Adaptive zero-Knowledge}]
    A non-interactive proof system $\NIZK$ is said to be adaptive zero-knowledge if there exists a PPT simulator $\mathrm{SIM} = \{\mathcal{S}_1,\mathcal{S}_2\}$ such that for every PPT adversary $\adv$, there exists a negligible function $\varepsilon$ such that:
    \[ Adv^{ZK}_{\NIZK,\adv,\mathcal{S}}(\lambda) \overset{\mathrm{def}}{=}  \bigg|\ \mathbb{P}\left[crs \leftarrow \ParGen(1^\lambda) :\adv^{\mathcal{P}(crs,\cdot,\cdot)}(crs) = 1 \right] - \] \[\mathbb{P}\left[(crs,\tau) \leftarrow \mathcal{S}_1(1^\lambda) :\adv^{\mathcal{S}(crs,\tau,\cdot,\cdot)}(crs) = 1 \right]\bigg| \leq \varepsilon(\lambda)\]
    where $\tau$ is a simulation trapdoor, $\mathcal{P}$ and $\mathcal{S}$ are two oracles that return $\bot$ if $R(x,w) \neq 1$ or, respectively, $\pi \leftarrow \Prove(crs,x,w)$ and $\pi \leftarrow \mathcal{S}_2(crs,\tau,x)$ otherwise.
     
    \begin{comment}

\end{comment}
\end{definition}

Furthermore, no adversary not knowing a valid witness should be able to convince the verifier with more than negligible probability.
This is modeled by the existence of an extractor which can extract a valid witness from any adversary that is able to generate valid proofs.
This should hold even if the adversary has previously seen simulated proofs on statements of its own choice.
\begin{definition}[\textbf{Weak Simulation-Sound Extractability}]
An adaptively zero-knowledge proof system $\NIZK$ achieves weak simulation sound extractability if there exists a PPT extractor $(\mathcal{S},\mathcal{E})$ such that for every PPT adversary $\adv$ we have:
\[Adv_{\NIZK,\adv}^{SIM-Sound}(\lambda) \overset{\mathrm{def}}{=}  \bigg|\ \mathbb{P}\left[(crs,\tau) \leftarrow \mathcal{S}_1(1^\lambda) :\adv(crs) = 1 \right]\]
\[ - \mathbb{P}\left[(crs,\tau,\zeta) \leftarrow \mathcal{S}(1^\lambda) :\adv(crs) = 1 \right]\bigg| \leq \varepsilon_1(\lambda)\]
    and,
    \[Adv_{\NIZK,\adv}^{Weak-Ext}(\lambda) \overset{\mathrm{def}}{=}  \mathbb{P}\left[\mathrm{Exp}^{Weak-Ext}_{\NIZK,\adv}(\lambda) = 1 \right] \leq \varepsilon_2(\lambda)\]
    where $\mathrm{Exp}^{Weak-Ext}_{\NIZK,\adv}$ is defined in \cref{wext}.    
\end{definition}

\begin{algorithm}[ht!]
\caption{$-$ $\mathrm{\mathbf{Exp}}^{Weak-Ext}_{\NIZK,\adv}(\lambda)$}\label{wext}
\begin{algorithmic}[1]
\State $(crs,\tau,\zeta) \longleftarrow \mathcal{S}(1^\lambda)$
\State $(x^*,\pi^*) \longleftarrow \adv^{\mathcal{S}(crs,\tau,\cdot)}(crs)$
\State $w \longleftarrow \mathcal{E}(crs,\tau,x^*,\pi^*)$
\If{$\Verify(crs,x^*,\pi^*) = 1 \ \land \ R(x^*,w)\neq 1 \land (x^*,\cdot)\notin \mathcal{Q_S}$}
    \Return $1$
\EndIf
\State \Return $0$\par \noindent
\hspace{3mm}
 Where $\mathcal{S}(crs,\tau,x) := \mathcal{S}_2(crs,\tau,x)$ and $\mathcal{Q_S}$ keeps track of queries and answers of $\mathcal{S}$.
\end{algorithmic}
\end{algorithm}

For further discussions, we refer, e.g., to Goldwasser et al.~\cite{STOC:GolMicRac85} and Derler and Slamanig~\cite{DCC:DerSla19}.

We adopt the Camenisch-Stadler framework~\cite{C:CamSta97} for representing proof goals. 
We write:
$$
  \pi\gets\ZKP\left[(\alpha,\beta,\gamma):Y=G^\alpha\cdot H^\beta ~\land~ Z=G^\alpha\cdot H^\gamma ~\land~ \gamma=\alpha\cdot\beta\right]
$$
to denote a NIZK demonstrating knowledge of the values $\alpha, \beta, \gamma$ that satisfy the relation on the right-hand side, where all values outside the parentheses are assumed to be public.

\subsubsection{Non Interactive Key Exchange}
A key exchange protocol allows two or more parties to securely establish a shared secret key over an insecure channel. 
A non-interactive key exchange (NIKE) protocol enables parties to compute a shared key without direct communication, typically using pre-exchanged public information. The first example of a NIKE scheme can be found in the seminal work Diffie and Hellman
\cite{DifHel76}. Our definition follows the work of Freire et al. \cite{PKC:FHKP13}, where we assume that the identities can implicitly be derived from the keys.

\begin{definition}[\textbf{Non Interactive Key Exchange protocol}]
    A non-interactive key exchange protocol $\NIKE$ is a tuple of PPT algorithms $\{\ParGen,\KeyGen, \ShareKey\}$ such that:
    \begin{itemize}
        \item $\ParGen(1^\lambda) \rightarrow \pp$: This algorithm takes a security parameter $\lambda$ and outputs public parameters $\pp$.
        \item $\KeyGen(\pp) \rightarrow (\sk,\pk)$: This algorithm takes public parameters $\pp$ and outputs a secret and public key $(\sk,\pk)$.
        \item $\ShareKey(\pk',\sk) \rightarrow k$: This algorithm takes a secret key $\sk$ and a public key $\pk'$ and outputs a shared key $k$.
    \end{itemize}
\end{definition}

An important feature of NIKEs is that they are symmetric in the sense that two users, each knowing their own secret key and the other's public key, will obtain the same shared key.
That is, if $(sk_0,pk_0)$ and $(sk_1,pk_1)$ are the keys of two users, it always holds that $\ShareKey(pk_1,sk_0)=\ShareKey(pk_0,sk_1)$.
To satisfy this requirement, $\ShareKey$ is usually designed as deterministic function, which we will also assume in the following. 

We say that a NIKE is secure, if no adversary not knowing the corresponding secret keys can decide which secret key was used to compute a specific shared key, even if the adversary can choose the public key and request arbitrary shared keys before. 

The following definition is targeted to our specific needs, and naturally satisfied by most existing constructions.
In particular, in contrast to previous work~\cite{EC:CasKilSho08,PKC:FHKP13}, we do not require any guarantees against rogue keys, as in our construction all public keys will be honestly generated, and thus only consider honestly generated keys in the experiment.
\begin{definition}[\textbf{Indistinguishability}]
    A NIKE protocol $\NIKE$ is said to be indistinguishable if for any PPT adversary $\adv$, there is a negligible function $\varepsilon$ such that:
        \[Adv^{IND-NIKE}_{\NIKE,\adv}(\lambda) \overset{\mathrm{def}}{=}  \bigg|\:\mathbb{P}\left[ \mathrm{Exp}_{\NIKE,\adv}^{\mathrm{IND-NIKE}}(\lambda) = 1 \right] - \frac{1}{2} \bigg| \leq \varepsilon(\lambda) \]
    where $\mathrm{Exp}_{\NIKE,\adv}^{\mathrm{IND-NIKE}}$ is the experiment described in \cref{ind-nike}.    

\begin{algorithm}[ht!]
\caption{$-$  $\mathrm{\mathbf{Exp}}_{\NIKE,\adv}^{\mathrm{IND-NIKE}}(\lambda)$}\label{ind-nike}
\begin{algorithmic}[1]
\State $\pp \longleftarrow \ParGen(1^\lambda)$
\State $b \overset{\$}{\leftarrow}\{0,1\}$
\State $\pk^*, pk_0,pk_1 \longleftarrow \adv^{\Oracle^{\NIKE},\Oracle^{GenU},\Oracle^{CorruptU}} (\pp)$\newline
Let $\sk_b$ be the secret key corresponding to $\pk_b$ 
\State $k^* \longleftarrow \ShareKey(\pk^*,\sk_b)$
\State $b' \longleftarrow \adv(\pp,k^*)$
\NoThen
\If{:}
    \State (a) $b = b'$ and
    \State (b) $\pk_0,\pk_1$ and $\pk^*$ have not been corrupted.
    \State \Return $1$
\EndIf
\State \Return $0$\par

\noindent Where the oracles behave as follows: 
\begin{description}
    \item[-] $\Oracle^{GenU}$ generates and stores a fresh key pair $(\sk,\pk)$ and outputs $\pk$.
    \item[-]  $\Oracle^{CorruptU}$, on input a previously generated $\pk$, outputs the corresponding $\sk$.
    \item[-] $\Oracle^{\NIKE}$, on input $\pk,\pk'$ that were previously generated, looks up $\sk'$ corresponding to $\pk'$ and returns $\ShareKey(\pk,\sk')$.
\end{description}
\end{algorithmic}
\end{algorithm}
\end{definition}

Finally, we additionally require that no two secret keys can result in the same public key.
Note that this property is naturally satisfied by most constructions over cyclic groups.

For the definition, remember that a function  $f : X \rightarrow Y$ is said to be injective if $f(x) = f(y)$ implies $x = y$ for all $x,y\in X$.
\begin{definition}[\textbf{Secret key to public key injective property}]
    A non-interactive key exchange protocol NIKE $\NIKE$ with secret key space $X$ and public key space $Y$ is secret-key-to-public-key injective if there exists an injective mapping $\mu : X \rightarrow Y$ such that for all $(\sk,\pk) \leftarrow \NIKE.\KeyGen(\pp)$ it holds that $\pk = \mu(\sk)$.
\end{definition}

\section{Framework for Delegatable Pseudonyms}\label{s:framework}
In the following we now define a framework for delegatable pseudo\-nym systems.
We therefore briefly recap the core functionalities on an informal level, based on the Austrian bPk system.
We then introduce the syntax for such systems, and finally provide formal and unambiguous definitions for the security properties of such schemes.

\subsection{Requirements}\label{s:reqs}
A first set of functional requirements immediately follows from the responsibilities of the Austrian \emph{Stammzahlregisterbehörde} as specified in the national eGovernment law: 

\begin{description}
    \item[Central authority pseudonym generation.]
      The central authority needs to be able to compute pseudonyms on behalf of all user from the first registration (e.g., birth or arrival in Austria) onward.
      \olive{It is not foreseen that users may opt out from the system.}
    \item[Pseudonym linking.]
      The central authority needs to be able to ``translate'' pseudonyms for a specific user from one domain to another to allow inter-organization linking when required (subject to clear regulations).
    \item[De-anonymization.]
      In case of abuse of anonymity, the central authority has to be able to re-identify the user's identity, without breaking the user's privacy in other domains.
\end{description}

From a security and privacy point of view, the following requirements are guaranteed by the actual system.
\begin{description}
    \item[Privacy/Anonymity.]
      Even colluding service providers shall not be able to decide by themselves whether two pseudonyms belong to the same user or not.
      That is, they should not be able to pool their information about a specific user, except for cases where the central authority previously translated pseudonyms from one domain to another.
    \item[Authenticity/Non-frameability.]
      Pseudonyms generated must be provably authentic, i.e., a malicious entity should not be able to compute pseudonyms for identities they do not own (or which do not exist at all).
      More precisely, even if a subset of users collude, they should not be able to issue a pseudonym that would be de-anonymized towards an identity not in this set of malicious users.
\end{description}
Finally, it is important to note that current system does not foresee the possibility to revoke a user or rotate keys.

While the current architecture as depicted in \cref{fig:bpks_old} satisfies all these (legally mandated) requirements, a system addressing the challenges mentioned in \cref{s:intro} additionally needs to satisfy requirements related to decentralization and increased metadata privacy: 
%For a better protection of privacy, we require our system to be decentralized, i.e. a user can make requests directly to the service provider with his own pseudonym and without going through the central authority. Hence, our new framework includes the following requirements :
\begin{description}
    \item[User pseudonym generation] In addition to the central authority, honest users should be able to compute their pseudonyms entirely locally themselves, \olive{without prejudice to the verification of authenticity by the service provider.}

    \item[SP pseudonym generation] Certain types of service providers (e.g., public agencies) may be able to also locally generate pseudonyms for arbitrary users, \olive{yet only within their own domain.}
\end{description}

The latter requirement is optional but might be interesting in case of unavailability of the central authority. We are also aware that this requires careful and delicate balancing between privacy and availability, as corrupt service providers may compute arbitrary pseudonyms and share this information with other entities.
We thus envision that service providers with this capability would  be subject to legal regulations (e.g., public agencies), and also that their secret key is secured, e.g., within a hardware security module (HSM) precisely logging each pseudonym computation, such that the behavior can be audited, e.g., by data protection authorities.

\olive{The resulting high-level flows of our architecture are depicted in \cref{fig:bpks_new} and \cref{sec:sequence-diagrams}:
Users can obtain their own key material from the central authority (\circled{1} -- \circled{3}), while trusted service providers optionally obtain their own secret key \circled{4}. 
Using the public key of a SP \circled{5}, users can now locally compute verifiable pseudonyms for SP \circled{6}, while trusted SPs can fetch a user public key \circled{7} to locally compute the pseudonym \circled{8}.
The authorities capabilities (e.g., linking pseudonyms \circled{9} or computing pseudonyms for less tech-savvy users) remain intact. }

\begin{figure}[bt!]
  \includegraphics[width=0.88\columnwidth]{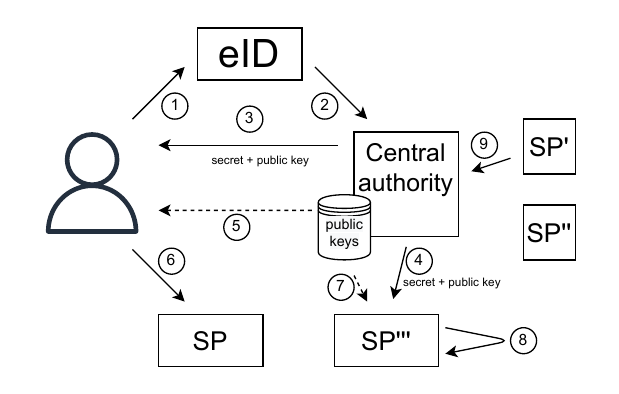}
  \caption{Logical information flow for the proposed \bpkplus architecture}
  \label{fig:bpks_new}
\end{figure}

\subsection{Syntax}\label{s:syntax}
For the initialization of the system, we require a global setup algorithm, defining common parameters, etc.
This algorithm only needs to be executed once for the entire system.
As usual, this algorithm is assumed to be carried out by a trusted party.
To reduce the required trust into this party, in practice one might implement the algorithm as a multi-party protocol\footnote{As demonstrated for example in \cite{AC:KMSV21}.}, with the use of random oracles for constructions in the random oracle model, or similar.
The resulting parameters are assumed to be input to all other algorithms, and will sometimes be omitted for notational convenience.
\begin{description}
    \item[$\ParGen(1^\lambda) \to \pp$] 
 is a probabilistic algorithm that, given security parameter $1^\lambda$, outputs public parameters $\pp$ of underlying primitives used in the scheme.
\end{description}

The following algorithms are executed by the central authority.
The first algorithm is only generated once, to establish a master key pair (corresponding, e.g., to the $3DES$ key in the current bPk system).
Subsequently, the central authority creates secret and public keys for all users and service providers.
Note that in many other privacy-enhancing protocols such as group signatures~\cite{EC:ChavHe91,DBLP:conf/IEEEares/KrennSS19} or attribute-based credentials~\cite{SAC:CKLMNP15}, the key generation happens on the user side.
Centralizing it in our modeling was an active design decision.
Firstly, pseudonyms must remain constant throughout life, making it difficult to make citizens responsible for backups of their secret key.
Secondly, and more importantly, keys must be available from birth, and there is no legal opportunity to opt-out.
Service provider keys are also generated by the authority, as only certain service providers (e.g., public authorities), should have access to the secret key at all, and if, it needs to be included, e.g., in secure hardware, cf. also \cref{s:reqs}:
\begin{description}
    \item[$\KeyGen(\pp) \to (\msk, \mpk)$] is a probabilistic algorithm that, given public parameters $\pp$ as input, outputs the master secret key $\msk$ and the master public key $\mpk$.
    \item[$\KeyGen_{user}(\msk) \to (\usk, \upk)$] is a probabilistic algorithm that, given the master secret key $\msk$, outputs a pair of public and private keys $(\usk,\upk)$ for the user .
    \item[$\KeyGen_{SP}(\msk) \to (\spsk, \sppk)$] is a probabilistic algorithm that, given the master secret key $\msk$, outputs a pair of public and private keys $(\spsk, \sppk)$ for the service provider.
\end{description}

Next, users need to be able to generate pseudonyms.
To achieve authenticity, they additionally generate a proof showing the well-formedness of the pseudonym, i.e., proving that it was derived from a honest user secret key.
This proof can then be verified by the service provider.
If pseudonyms are computed by the service provider (for those that have access to the secret key corresponding to their public key), obviously no such proof is needed.
To keep the model as light-weight as possible, we do not add an additional interface for pseudonym generation through the central authority, as it can simply leverage the service prover's interface as well.
\begin{description}
  \item[$\NymGen_{user}(\usk,\mpk,\sppk) \to (\nym,\pi)$] is an algorithm that, given $\usk,\upk,\sppk$ and  $\mpk$, outputs the user pseudonym $\nym$ and validity proof $\pi$. 
  \item[$\NymVerify(\mpk,\sppk,\nym,\pi) \to 0/1$] is a deterministic algorithm that, given the master public key $\msk$, a service provider public key $\sppk$, a pseudonym $\nym$, a proof $\pi$, outputs $1$ or $0$ whether $\pi$ is valid for $\nym$ or not.
  \item[$\NymGen_{SP}(\spsk,\mpk,\upk) \to \nym$] is a deterministic algorithm that, given the secret key of the service provider $\spsk$ and the public key of the user $\upk$, outputs the user pseudonym $\nym$ specific to a service provider.
\end{description}

Finally, in case of abuse, the central authority has to be able to re-identify the user from a pseudonym and proof.
Also, this allows the central authority to translate a pseudonym from one domain to another, by first opening a given pseudonym, and then computing the new pseudonym for a different domain for the same user.
\begin{description}
  \item[$\Open(\pi, \msk) \to \upk$] is a deterministic algorithm that, given a proof $\pi$ and the master secret key $\msk$, outputs the user public key $\upk$.
\end{description}

We want to remark that this final algorithm might in fact not be necessary for cross-domain linking if the approach currently used in the Austrian bPk system is followed. In that system, when a pseudonym for a given domain is known, the underlying user can only be identified either by brute-forcing all possible users or by consulting a lookup table containing all issued pseudonyms, since the pseudonym generation relies on a one-way function.

\subsection{Security Framework}
From a security point of view, we expect three main properties:
completeness, saying that honest users can always authenticate;
non-frameability, capturing that pseudonyms need to be verifiable and guarantee authenticity;
and anonymity, saying that pseudonyms should not leak any information about a user's identity, including especially also unlinkability.

We define these properties formally in the following.

\subsubsection{Correctness}
We omit a formal definition of correctness.
Intuitively, it ensures that honest users will always be able to successfully present pseudonyms to honest service providers, and that the (trusted) central authority will always be able to link or re-identify honest users.

Furthermore, as illustrated in \cref{fig:correctness}, service providers and users should always obtain the same pseudonym for a given user (i.e., knowing $\usk$ and $\sppk$ should result in the same $\nym$ as starting from $\upk$ and $\spsk$).

\begin{figure}
    \centering
\[\begin{tikzcd}[column sep=2.8cm]
\text{CA}  \arrow{r}{(\usk,\sppk)} \arrow[swap]{d}{(\spsk,\upk)} & \text{U} \arrow{d}{\NymGen_{user}} \\
\text{SP} \arrow{r}{\NymGen_{SP}} & \nym
\end{tikzcd}
\]
    \caption{Correctness of pseudonym generation}
    \label{fig:correctness}
\end{figure}
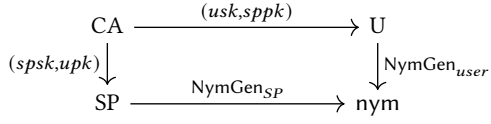

\subsubsection{Non-Frameability}
  As discussed earlier, authenticity of user-generated pseudonyms is of high importance due to the potential legal implications of formal authentications.
  This property is captured by our notion of non-frameability, which ensures that an honest user cannot be falsely accused of having generated a pseudonym and corresponding proof.
  This should even hold if all other users collude, in which case they still cannot produce a valid pseudonym and proof that falsely implicates an innocent member. 
  In contrast to the non-frameability notion, e.g., in group signatures~\cite{ACNS:BCCGG16,DBLP:conf/IEEEares/KrennSS19}, the central authority must not collude, as our framework does contain secret key material solely held by the user (cf. the discussion in \cref{s:syntax}).
  Additionally, our notion also captures the basic unforgeability notion, which says that colluding users also cannot generate a valid pseudonym and proof which do not link to \emph{any} existing user in the system.

  In the experiment, $\mathrm{N}$ denotes the set of queried pseudonyms. 
Let $KP_u$ be the set of all created pairs of user public and private keys generated by the central authority. We also define the subsets $\mathrm{PK_u}$ (resp. $\mathrm{SK_u}$) of $KP_u$ as the set of honestly generated user public (resp. private) keys. Let $\mathrm{PK_{sp}}$, $\mathrm{SK_{sp}}$, and $\mathrm{KP_{sp}}$ be the respective sets of service provider keys generated by the central authority. Finally, $\mathrm{C_{u}}$ denotes the set of corrupted users.

\begin{definition}[\textbf{Non-Frameability}]
    A delegateable pseudonym system is said to be non-frameable, if for every PPT adversary $\adv$ there exists a negligible function $\varepsilon(.)$ such that: 
    $$
    \mathbb{P}\left[ \mathrm{Exp}_{\adv}^{\mathrm{Non-Frameability}}(\lambda) = 1 \right] \leq \varepsilon(\secpar),
    $$
    where $\mathrm{Exp}_{\adv}^{\mathrm{Non-Frameability}}(\lambda)$ is defined in \cref{exp:non-frameability}.  
\end{definition}

\begin{algorithm}
\caption{$-$ $\mathrm{\mathbf{Exp}}_{\adv}^{\mathrm{Non-Frameability}}(\lambda)$}\label{alg:cap}
\begin{algorithmic}[1]
\State $\mathrm{N}, \mathrm{KP_u}, \mathrm{PK_u}, \mathrm{SK_{u}}, \mathrm{PK_{sp}}, \mathrm{SK_{sp}}, \mathrm{C_u} \leftarrow 	\emptyset$
\State $\pp \longleftarrow \ParGen(1^\lambda)$
\State $(\msk, \mpk) \longleftarrow \KeyGen(\pp)$
\State $(\sppk,\nym^*, \pi^*) \longleftarrow \adv^{\Oracle^{nym}, \Oracle^{GenU}, \Oracle^{GenSP}, \Oracle^{CorrupU}}(\pp,\mpk)$
\NoThen
\If{:} \COMMENT{Adversary outputs a pseudonym of a non-existing user}
    \State (a) $\NymVerify(\nym^*,\pi^*,\mpk, \sppk) = 1$ %// $\pi^*$ is valid 
    \State (b) For $\upk^*=\Open(\pi,\msk), \ \nexists (\upk^*,\usk^*)\in \mathrm{KP_u}$%//$c^*$ opens to a non-existing user
    \State \Return 1
\EndIf
\If{:} \COMMENT{Adversary forges a fresh pseudonym for an honest user}
    \State (a) $\NymVerify(\nym^*,\pi^*,\mpk, \sppk) = 1$ %// $\pi^*$ is valid
    \State (b) For $\upk^*=\Open(\pi,\msk) ,\exists(\upk^*,\usk^*)\in\mathrm{KP_u}$  %//$c^*$ opens to a real user (i.e., honestly generated - i.e., not to $\bot$ or new pk)
    \State (c) $(\upk^*,\sppk) \notin \mathrm{N}$ %//no replay
    \State (d) $\usk^* \notin C_{u}$ %//$\usk^*$ is not corrupted 
    \State \Return 1
\EndIf
\State \Return 0
\end{algorithmic}\label{exp:non-frameability}
\end{algorithm}

In the security experiment, the adversary is given access to a number of oracles.
Intuitively, upon the adversary's request, $\Oracle^{GenU}$ creates a new honest user, and hands the corresponding public key $\upk$ to $\adv$.
Similarly, the adversary may request the creation of additional service providers through a call to $\Oracle^{GenSP}$, again obtaining only the corresponding public key $\sppk$.
For any pair of users and service providers, the adversary may now request the corresponding user pseudonym, together with a proof of well-formedness, through $\Oracle^{Nym}$;
in case that the indicated service provider or user does not exist, the oracle returns $\bot$.
Finally, $\adv$ can corrupt users, thereby obtaining the corresponding secret key $\usk$, through $\Oracle^{CorruptU}$.

We now formally define these oracles in \cref{o:genu,o:gensp,o:nym,o:corruptu}.
Note that no oracle for corrupting service providers exists:
as discussed earlier, untrusted service providers do not obtain their secret key, so that they are unable to generate pseudonyms by themselves.
Trusted (e.g., publicly hosted) service providers should only receive their $\spsk$ within secure hardware, such that key extraction is not possible, either.

\begin{oracle}
\caption{User generation oracle $\Oracle^{GenU}$}\label{o:genu}
\begin{algorithmic}[1]
\State $(\usk,\upk) = \KeyGen_{user}(msk)$ 
\State $\mathrm{PK_u} = \mathrm{PK_u} \cup \{\upk\}$
\State $\mathrm{SK_u} = \mathrm{SK_u} \cup \{\usk\}$
\State $\mathrm{KP_u} = \mathrm{KP_u} \cup \{(\usk,\upk)\}$
\State \Return $\upk$
\end{algorithmic}
\end{oracle}

\begin{oracle}
\caption{Service provider generation oracle $\Oracle^{GenSP}$}\label{o:gensp}
\begin{algorithmic}[1]
\Ensure Service provider public key $\sppk$
\State $(\spsk,\sppk) = \KeyGen_{user}(msk)$ 
\State $\mathrm{PK_{sp}} = \mathrm{PK_{sp}} \cup \{\sppk\}$
\State $\mathrm{SK_{sp}} = \mathrm{SK_{sp}} \cup \{\spsk\}$
\State $\mathrm{KP_{sp}} = \mathrm{KP_{sp}} \cup \{(\spsk,\sppk)\}$
\State \Return $\sppk$
\end{algorithmic}
\end{oracle} 

\begin{oracle}
\caption{Pseudonym generation oracle $\Oracle^{Nym}$}\label{o:nym}
\begin{algorithmic}[1]
\Require Service provider public key $\sppk$, user public key $\upk$
\Ensure pseudonym $\nym$ and proof $\pi$. 
\If{$\sppk \notin \mathrm{PK_{sp}}$ or $\upk \notin \mathrm{PK_u}$:}
    \Return $\perp$
\EndIf
%\State Find $spsk \in \mathrm{SK_{sp}}$ s.t $\spsk,\sppk = \KeyGen_{}(\pp)$
\State Find $(\usk,\upk) \in \mathrm{KP_u}$
\State Compute $(\nym,\pi) = \NymGen_{user}(\usk,\mpk,\sppk,\mpk)$
\State $N \leftarrow N \cup \left\{(\upk,\sppk)\right\}$
%\State $N \leftarrow N \cup \left\{(nym,\Pi,c)\right\}$
\State \Return $\nym, \pi$
\end{algorithmic}
\end{oracle}

\begin{oracle}[ht!]
\caption{Corrupted user oracle $\Oracle^{CorruptU}$}\label{o:corruptu}
\begin{algorithmic}[1]
\Require User public key $\upk$
\Ensure User secret key $\usk$
\If{$upk \notin \mathrm{PK_u}$}
    \Return $\perp$
\EndIf
\State Find $(\usk,\upk) \in \mathrm{KP_u}$
\State $C_u \leftarrow C_u \cup \left\{(\usk,\upk)\right\}$ 
\State \Return $\usk$
\end{algorithmic}
\end{oracle}

\subsubsection{Anonymity}
In the following we now formalize the main privacy property of delegatable pseudonym systems.
In a nutshell, the goal is that no adversary is able to decide significantly better than random guessing whether by which of two (adversarially chosen) users a specific pseudonym was created.

\begin{definition}[\textbf{Anonymity}]
    A delegatable pseudonym system is said to be anonymous, if for every PPT adversary $\adv$ there exists a negligible function $\varepsilon(.)$ such that: 
    $$
    \left|\mathbb{P}\left[ \mathrm{Exp}_{\adv}^{\mathrm{Anonymity}}(\lambda) = 1 \right] -\frac{1}{2} \right|\leq \varepsilon(\secpar),
    $$
    where $\mathrm{Exp}_{\adv}^{\mathrm{Anonymity}}(\lambda)$ is defined in \cref{exp:anonymity}.  
\end{definition}

Informally, anonymity is modeled through a left-or-right (LoR) oracle, which the adversary can call to generate pseudonyms for arbitrary service providers, and which will always return the pseudonym either for the first or for the second input.
The adversary wins if it can decide which is the case, under the following constraints:
$(b)$ for the same service provider, the adversary must not query inconsistent inputs, e.g., $(\upk_1,\upk_2)$ and then $(\upk_1,\upk_3)$, as this would trivially allow the adversary to decide which output it receives;
$(c)$ similarly, it must not request an honest pseudonym for a service provider that it also used in the LoR oracle, e.g., $(\upk_1,\upk_2)$ to LoR and also obtain a pseudonym for $\upk_1$, as this trivially allows to guess the output it receives;
$(d)$ it must not send $\upk$'s that it (previously or later on) corrupted to the LoR oracle, as again it could simply check for the output of the LoR oracle; and finally
$(e)$ it must not request pseudonyms for corrupted service providers or pseudonyms. While the latter is clear similarly to $(d)$, it is rather an artifact from our proof strategy which however does not impose real limitations in practice, as we (as discussed before) assume that $\spsk$ is never given to the adversary in the plain. 

\begin{algorithm}[ht!]
\caption{$-$ $\mathrm{\mathbf{Exp}}_{\adv}^{\mathrm{Anonymity}}(\lambda)$}\label{exp:anonymity}
\begin{algorithmic}[1]
  \State $b \overset{\$}{\leftarrow} \{0,1\}$
\State $\mathrm{LR}, \mathrm{N},\mathrm{KP_u}, \mathrm{PK_u}, \mathrm{SK_{u}}, \mathrm{PK_{sp}}, \mathrm{SK_{sp}},\mathrm{C_u}, \mathrm{C_{sp}} \longleftarrow \emptyset$
\State $\pp \longleftarrow \ParGen(1^\lambda)$
\State $(\msk, \mpk) \longleftarrow \KeyGen(\pp)$

\State $b' \longleftarrow \adv^{\Oracle^{LoR},\Oracle^{GenU},\Oracle^{GenSP}, \Oracle^{Nym}, \Oracle^{CorrupU}, \Oracle^{CorruptSP}(\pp,\mpk)}$
\NoThen
\If{:}
    \State (a) $b' = b$
    
    \State (b) $\forall \mathrm{A} = \left\{\upk_0,\upk_1\right\}, \mathrm{A'} = \left\{\upk'_0,\upk'_1\right\}$ s.t $(A,sppk), (A',sppk)\in \mathrm{LR},$ we have $ A = A'$ or $A \cap A' = \emptyset $
    
    \State (c) $\forall (\upk,\sppk) \in \mathrm{N}, \upk' \in \mathrm{PK_u} : \nexists(\left\{\upk,\upk'\right\}, \sppk)$ or $(\left\{\upk',\upk\right\}, \sppk) \in \mathrm{LR}$ 
    
    \State (d) $\forall (\left\{\upk,\upk'\right\}, \sppk) \in \mathrm{LR}, \upk,\upk' \notin \mathrm{C_u}$ and $\sppk \notin \mathrm{C_{SP}}$
    
    \State (e) $\forall (\upk,\sppk) \in \mathrm{N}, \upk \notin \mathrm{C_u}$ and $ \sppk \notin \mathrm{C_{sp}}$ 
    \State \Return 1
\EndIf
\State \Return 0
\end{algorithmic}
\end{algorithm}

The definition of the experiment depends on two additional oracles beyond those used in the previous section.
That is, $\Oracle^{LoR}$ always returns the $b^{\text{th}}$ pseudonym and proof for two inputs $(\upk,\upk')$ for a given service provider.
Finally, service providers can be corrupted through the $\Oracle^{CorruptSP}$ oracle.
We formally define these oracles in \cref{o:lor,o:corruptsp}.

\begin{oracle}
\caption{Left or Right oracle $\Oracle^{LoR}$}\label{o:lor}
\begin{algorithmic}[1]
\Require Set of two public keys $\left\{\upk_0,\upk_1\right\}$ and service provider public key $\sppk$ 
\Ensure Pseudonym $\nym_b$
\If{$upk_0 \notin \mathrm{PK_u}$ or $upk_1 \notin \mathrm{PK_u}$ or $sppk \notin \mathrm{PK_{sp}}$}
    \Return $\perp$
\EndIf

\State Find $(\usk_i,\upk_i) \in \mathrm{KP_u}$ %s.t $\usk_i,\upk_i = \KeyGen_{user}(\msk)$ 
for $i = 0,1$ 
\State Compute $(\nym_i,\pi_i) = \NymGen_{user}
(\usk_i,\mpk,\sppk)$ for $i = 0,1$ 
\State $LR \leftarrow LR \cup \left\{(\left\{\upk_0,\upk_1\right\},\sppk)\right\}$
\State \Return $\nym_b,\pi_b$
\end{algorithmic}
\end{oracle}

\begin{oracle}
\caption{Corrupted service provider oracle $\Oracle^{CorruptSP}$}\label{o:corruptsp}
\begin{algorithmic}[1]
\Require Service provider public key $\sppk$
\Ensure Service provider secret key $\spsk$
\If{$sppk \notin \mathrm{PK_{sp}}$}
    \Return $\perp$
\EndIf
\State Find $(\spsk,\sppk) \in \mathrm{KP_{sp}}$ %s.t $(\spsk,\sppk) = \KeyGen_{SP}(\msk)$
\State $C_{sp} \leftarrow C_{sp} \cup \left\{(\spsk,\sppk)\right\}$ 
\State \Return $spsk$
\end{algorithmic}
\end{oracle}

% \section{Generic Construction}\label{s:generic}
% In the following we now introduce a generic construction, and provide rigorous proofs that it achieves non-frameability and anonymity.

\section{Construction} \label{s:generic}
In the following we provide a blackbox construction from a set of cryptographic primitives as defined in \cref{s:preliminaries} and provide rigorous proofs that it achieves non-frameability and anonymity.
Therefore, let $\Sig$ be a signature scheme, $\Enc$ be an encryption scheme, $\NIKE$ be a non-interactive key exchange protocol with secret key to public key injective property, and $\NIZK$ be a non-interactive weakly simulation-sound proof system.
Without loss of generality, we assume that the schemes are compatible in the sense that the message spaces of $\Sig$ and $\Enc$ include  the public key space of $\NIKE$.   

The intuition of our construction is now as follows.
The central authority generates the NIKE keys for all users and service providers, and pseudonyms are simply shared keys derived from the NIKE between a service provider and a user, which immediately yields commutativity, i.e., they will always obtain the same pseudonym.

In order to ensure that user-generated pseudonyms are authentic, the user does not only receive their secret key of the NIKE, but in addition a signature from the central authority on the public key.
Now, when the user generates a pseudonym $\nym$ they also compute a proof $\pi$ that they used the service provider's public key and a user secret key for which they also know a signature on the corresponding user public key. 

Finally, to enable the central authority to trace users, the user in addition adds an encrypted version of their public key to the proof goal.

A detailed and formal specification of ${\bpkplus}_{\Enc,\Sig,\NIKE,\NIZK}$ is given in Construction ~\ref{Framework}.

\begin{construction}[t!]
\caption{$\bpkplus_{\Enc,\Sig,\NIKE,\NIZK}$ with public-key key encryption scheme $\Enc$, signature scheme $\Sig$, NIKE $\NIKE$ and NIZK $\NIZK$ such that the public-key space of $\NIKE$ is (a subset of) the message spaces of $\Enc$ and $\Sigma$. We also define $\mathcal{R}$ as the set in which the random used in the encryption process of $\Enc$ is sampled.}\label{Framework}

\begin{description}[leftmargin=20pt,labelindent=0pt]
    \item $\ParGen(1^\lambda)$ :\begin{algorithmic}[1]
        \State $\pp_\Sig \gets \Sig.\ParGen(1^\lambda)$
        \State $\pp_{\Enc} \gets \Enc.\ParGen(1^\lambda)$
        \State $\pp_{\mathrm{\NIZK}} \gets \NIZK.\ParGen(1^\lambda)$
        \State $\pp_{\mathrm{\NIKE}} \gets \NIKE.\ParGen(1^\lambda)$
        \State \Return $\pp = \{\pp_\Sig,\pp_{\Enc},\pp_\NIZK,\pp_{\mathrm{\NIKE}}\}$
    \end{algorithmic}
    
    \item $\KeyGen(\pp)$ : \begin{algorithmic}[1]
        \State$(\msk_\Sig,\mpk_\Sig) \gets \Sig.\KeyGen(\pp_\Sig)$  
        \State $(\msk_{\Enc},\mpk_{\Enc}) \gets \Enc.\KeyGen(\pp_\Enc)$ 
        \State $\msk \gets (\msk_\Sig,\msk_{\Enc})$ 
        \State $\mpk \gets (\mpk_\Sig,\mpk_{\Enc})$ 
        \State \Return $\msk,\mpk$
    \end{algorithmic}
    \item $\KeyGen_{user}(\msk)$ : \begin{algorithmic}[1]
        \State $(\usk',\upk) \gets \NIKE.\KeyGen(\pp_{\mathrm{\NIKE}})$
        \State $\sigma \gets \Sig.\Sign(\upk,\msk_{\Sig})$
        \State \Return $\usk = (\usk',\sigma),\upk$
    \end{algorithmic}
    
    \item $\KeyGen_{SP}(\msk)$ : \begin{algorithmic}[1]
        \State$(\spsk,\sppk) \gets \NIKE.\KeyGen(\pp_{\NIKE})$
        \State \Return $\spsk ,\sppk$
    \end{algorithmic}
    
    \item $\NymGen_{user}(\usk,\mpk,\sppk)$ :\begin{algorithmic}[1]
        \State Parse $\usk$ as $(\usk',\sigma)$
        \State $\nym \gets \NIKE.\ShareKey(\sppk,\usk')$
        \State $c \gets \Enc.\Encrypt(\mpk,\upk;r)\quad r\getsr \mathcal{R}$ 
        \State $\pi \gets (c,\NIZK.\Prove[(\usk',\upk,\sigma,r) :$ 
            \begin{align*} 
                &\ \Sig.\Verify(\mpk,\sigma,\upk) = 1\\ & \land c = \Enc.\Encrypt(\mpk,\upk;r)\\ &\land \upk = \mu(\usk')\\ &\land \nym = \NIKE.\ShareKey(\sppk,\usk')])
            \end{align*}
        \State \Return $\nym,\pi$
    \end{algorithmic}
    
    \item $\NymGen_{SP}(\spsk,\mpk,\upk)$ : \begin{algorithmic}[1]
        \State \Return $\nym \gets\NIKE.\ShareKey(\upk,\spsk)$
    \end{algorithmic}
    
    \item $\NymVerify(\mpk,\sppk,\nym,\pi)$ :\begin{algorithmic}[1]    
        \State \Return $\NIZK.\Verify(\pp_\NIZK,(\mpk,\sppk,\nym),\pi)$
    \end{algorithmic}
    
    \item $\Open(\pi,\msk)$ :\begin{algorithmic}[1]
        \State Parse $\pi$ as $(c,\pi')$
        \State \Return $ \upk \gets \Enc.\Decrypt(\msk,c)$
    \end{algorithmic}
\end{description}

\end{construction}

We now obtain the following core results regarding the security of our generic construction.

\begin{theorem}\label{thm:nonframeability}
    Let $\Enc$ be a public-key encryption scheme, $\Sig$ be an EUF-CMA signature scheme, $\NIKE$ be a NIKE scheme satisfying indistinguishability and secret to public key injective properties and $\NIZK$ be a ZKPoK system. Then, $\bpkplus_{\Enc,\Sig,\NIKE,\NIZK}$ achieves non-frameability against an adaptive PPT adversary $\adv$.
\end{theorem}
Intuitively, to see non-frameability, note that an adversary would either have to break the proof-of-knowledge guarantees of the NIZK (to prove that it did a valid computation even though it did not) or forge a signature (in order to introduce a new, non-existing user to the system), yet both these cases are excluded by the guarantees of the building blocks.
For the full proof, we refer to \cref{s:proof:nonframeability}.

\begin{theorem}\label{thm:anonymity}
     Let $\Enc$ be an IND-CPA public-key encryption scheme, $\Sig$ be an EUF-CMA signature scheme, $\NIKE$ be a NIKE scheme satisfying indistinguishability and secret to public key injective properties and $\NIZK$ be a ZKPoK system. Then, $\bpkplus_{\Enc,\Sig,\NIKE,\NIZK}$ achieves anonymity against an adaptive PPT adversary $\adv$.
\end{theorem}
Anonymity basically follows from the fact that user prove their statements in zero-knowledge, and that the encryption of their identity does not leak information to a computationally bounded adversary either. 
Furthermore, the indistinguishability property of the NIKE guarantees that also known pseudonyms cannot be used to link users to fresh $\sppk$'s.
For the full proof, we refer to \cref{s:proof:anonymity}.

\section{Instantiation}\label{s:instantiation}

In the following, we consider a pairing-based instantiation of our generic \bpkplus\ construction with a concrete selection of compatible schemes. 

\subsection{Building Blocks}
The main constraint for our construction is the choice of a signature scheme which can efficiently be combined with a NIZK to support the proof of knowledge of a valid signature from the central authority. We are inspired by the constructions of \cite{CCS:CamDriDub17,CANS:BEKRS21} which are based on Groth's structure-preserving signatures~\cite{AC:Groth15}. We only need to sign a single element in one of the source groups and thus give a presentation of the signature scheme adapted to our needs\footnote{We present the variant that signs a group element in $\GG_1$. By switching the role of $\GG_1$ and $\GG_2$ in the construction, the signature scheme would allow signing group elements in $\GG_2$.}:
\begin{itemize}
	\item $\Groth_1.\ParGen(1^\secpar)$ generates public parameters of a bilinear group $(\GG_1,\GG_2,\GG_T,e,p,G,\hat{G})$ of prime order $p$, where $G\in\GG_1$ and $\hat{G}\in\GG_2$ are generators. It additionally outputs an element $Y\getsr\GG_1$.
  \item $\Groth_1.\KeyGen()$ samples $\sk\gets\ZZ_p^*$ and sets $\pk\gets\hat{G}^\sk$.
  \item $\Groth_1.\Sign(\sk,\msg)$ samples $r\gets\ZZ_p^*$ and computes a signature $\sigma=(\hat{R},S,T)=(\hat{G}^r,(Y\cdot G^\sk)^{1/r},(Y^\sk\cdot\msg)^{1/r})$.
  \item $\Groth_1.\Rand(\sigma)$ rerandomizes a valid signature $\sigma$ by sampling $r'\gets\ZZ_p^*$ and outputting a randomized signature $\sigma'=(\hat{R}',S',T')=(\hat{R}^{r'},S^{1/r'},T^{1/r'})$ on the same message.
  \item $\Groth_1.\Verify(\pk,\sigma,\msg)$ outputs $1$ if and only if it holds that $e(S,\hat{R})=e(Y,\hat{G})\cdot e(G,\pk)$ and $e(T,\hat{R})=e(Y,\pk)\cdot e(\msg,\hat{G})$.      
\end{itemize}

For the NIKE we need to also select schemes with public keys in $\GG_1$. Under the XDH assumption, note that we can instantiate the Diffie-Hellman~\cite{DifHel76} NIKE in $\GG_1$. Hence, we obtain the following NIKE:
\begin{itemize}
  \item $\DHN.\ParGen(1^\secpar)$ sets $\pp$ to $\GG_1$ and $H \getsr \GG_1$. Returns $\pp$.
  \item $\DHN.\KeyGen()$ samples $\sk \getsr \ZZ_p^*$ and sets $\pk \gets H^\sk$.
  \item $\DHN.\ShareKey(\pk', \sk)$ computes $\shk \gets \pk'^\sk$.
\end{itemize} 

Similarly, under the XDH assumption, ElGamal encryption~\cite{C:ElGamal84} provides a secure public-key encryption scheme in $\GG_1$ with message space $\GG_1$:
\begin{itemize}
  \item $\ElGamal.\ParGen(1^\secpar)$ sets $\pp$ to $\GG_1$ and $K \getsr \GG_1$. Returns $\pp$.
  \item $\ElGamal.\KeyGen()$ samples $\sk \getsr \ZZ_p^*$ and sets $\pk \gets K^\sk$.
  \item $\ElGamal.\Encrypt(\pk, \msg)$ computes $r \getsr \ZZ_p^*$, return $(K^r,\allowbreak \pk^r \msg)$.
  \item $\ElGamal.\Decrypt(\sk, (c_1, c_2))$ returns $c_2 c_1^{-\sk}$.
\end{itemize}

\subsection{Protocol Specification}
Putting the above building blocks into the generic construction presented in \cref{s:generic}, we now obtain the concrete instantiation described in \cref{const:concrete}.

\begin{construction}[t!]
\caption{Instantiating $\bpkplus$ with ElGamal encryption, Groth signatures, Diffie-Hellman key exchange, and Schnorr proofs.}\label{const:concrete}

\begin{description}[leftmargin=20pt,labelindent=0pt]
    \item $\ParGen(1^\lambda)$ :\begin{algorithmic}[1]
        \State Generate public parameters of a bilinear group $\pp'= (\GG_1,\GG_2,\GG_T,e,p,G,\hat{G})$ of prime order $p$, where $G\in\GG_1$ and $\hat{G}\in\GG_2$ are generators
        \State Sample $Y\getsr\GG_1$, $H \getsr \GG_1$, and $K \getsr \GG_1$.
        \State \Return $\pp = (\pp',Y,H,K)$.
    \end{algorithmic}
    
    \item $\KeyGen(\pp)$ : \begin{algorithmic}[1]
        \State  $(\msk_{\Groth_1}, \mpk_{\Groth_1}) \gets \Groth_1.\KeyGen()$ 
        \State $(\msk_\ElGamal, \mpk_\ElGamal) \gets \ElGamal.\KeyGen()$
        \State \Return $\msk \gets (\msk_{\Groth_1}, \msk_\ElGamal)$ and $\mpk \gets (\mpk_{\Groth_1}, \mpk_\ElGamal)$.
    \end{algorithmic}
    
    \item $\KeyGen_{user}(\pp,\msk)$ : \begin{algorithmic}[1]
        \State $(\usk', \upk) \gets \DHN.\KeyGen() = (\usk',H^{\usk'})$ for $\usk'\gets\ZZ_p^*$
        \State $\sigma \gets \Groth_1.\Sign(\upk, \msk_{\Groth_1})=(\hat{R},S,T)=(\hat{G}^r,(Y\cdot G^{\msk_{\Groth_1}})^{1/r},(Y^{\msk_{\Groth_1}}\cdot\upk)^{1/r})$ for $r\gets\ZZ_p^*$
        \State \Return $\usk = (\usk', \sigma)$ and $\upk$.
    \end{algorithmic}
    
    \item $\KeyGen_{SP}(\pp,\msk)$ : \begin{algorithmic}[1]
        \State \Return $\DHN.\KeyGen()= (\spsk,H^\spsk)$ for $\spsk\gets\ZZ_p^*$
    \end{algorithmic}
    
    \item $\NymGen_{user}(\usk,\mpk,\sppk)$ :\begin{algorithmic}[1]
        \State $\nym \gets \DHN.\ShareKey(\sppk, \usk')=\sppk^{\usk'}$
        \State $c = (c_1, c_2) \getsr \ElGamal.\Encrypt(\mpk_\ElGamal, \upk; r) = (K^r, \mpk_\ElGamal^r \upk)$ with $r \getsr \ZZ_p^*$
        \State $\sigma' \gets  \Groth_1.\Rand(\sigma; r') = (\hat{R}', S', T') = (\hat{R}^{r'}, S^{1/r'}, T^{1/r})$  with $r' \getsr \ZZ_p^*$.
        \State $(\hat{R}'',S'',T'')\gets (\hat{R}',S'^{1/\alpha},T'^{1/\beta})$ for $\alpha, \beta\getsr\ZZ_p^*$
        \State $\upk' \gets \upk^{1/s}$ for $s \getsr \ZZ_p^*$
        \State Produce a Schnorr-style proof:
    \begin{align*}
      \pi' \gets \NIZK.&\Prove[(r, s, \usk', \sigma'): \\
      % signature verification
      &e(S'',\hat{R}'')^\alpha=e(Y,\hat{G})\cdot e(G,\mpk_{\Groth_1}) ~\land\\
      &e(T'',\hat{R}'')^\beta =e(Y,\mpk_{\Groth_1})\cdot e(\upk',\hat{G})^s ~\land\\
      % encryption
      &c_1 = K^r ~\land c_2 = \upk'^s\cdot\mpk_{\ElGamal}^r ~\land\\
      % well-formed keys
      &\upk'^{s} = H^{\usk'} ~\land
      % ShareKey
      \nym = \sppk^{\usk'}]
    \end{align*} 
        \State $\pi \gets (\pi', c, \hat{R}'', S'', T'', \upk')$
        \State \Return $(\nym, \pi)$
    \end{algorithmic}
    
    \item $\NymGen_{SP}(\spsk,\mpk,\upk)$ : \begin{algorithmic}[1]
        \State \Return $\DHN.\ShareKey(\upk, \spsk)$.
    \end{algorithmic}
    
    \item $\NymVerify(\mpk,\sppk,\nym,\pi)$ :\begin{algorithmic}[1]    
        \State Parse $\pi$ as $(\pi', c, \hat{R}'', S'', T'', \upk')$. 
         \State \Return the verification result of $\pi'$
    \end{algorithmic}
    
    \item $\Open(\pi,\msk)$ :\begin{algorithmic}[1]
        \State Parse $\pi$ as $(\pi', c, \hat{R}'', S'', T'', \upk')$
        \State \Return $\upk \gets \ElGamal.\Decrypt(\msk_\ElGamal, c)$ 
    \end{algorithmic}
\end{description}
\end{construction}

Note that the proof goal could be slightly simplified, by replacing $\upk'^s$ by $H^{\usk'}$ in all equations, and removing the explicit term.
By proving the relation $\upk=\mu(\usk')$ implicitly, one can reduce the computational complexity of the proof goal by one discrete logarithm, and also reduce the communication complexity by removing $\upk'$ from $\pi$.
However, we refrain from this slight optimization to allow for a direct mapping of the instantiation to our generic construction. 

Note that Schnorr-style proofs made non-interactive using the Fiat-Shamir transform~\cite{C:FiaSha86} induce a simulation-sound extractable NIZK provided that the underlying $\Sigma$-protocol provides quasi-unique responses~\cite{INDOCRYPT:FKMV12} and the statement is used as part of the challenge generation~\cite{AC:BerPerWar12}. As the statement is a conjunction of proofs of knowledge of discrete logarithms, the responses are quasi-unique. Alternatively, the proof system could also be instantiated with Groth-Sahai~\cite{EC:GroSah08} to prove the pairing equations. 

\subsection{Efficiency Evaluation}
The above pairing-based instantiation was implemented\footnote{\url{https://anonymous.4open.science/r/bpkplus-4CE0/}} in Rust 1.85 using the BLS12-381 bilinear pairing via the \texttt{ark-bls12-381} crate\footnote{\url{https://crates.io/crates/ark_bls12_381}, version 0.5.0.}. To demonstrate the efficiency of our construction, we have benchmarked the implementation on an  Intel Core i7-1265U with 16 GB of RAM running Ubuntu 24.04. Since all algorithms except $\NymGen$ and $\NymVerify$ consist of at most 3 efficient group operations, we present the results of $\NymGen$ and $\NymVerify$ in \cref{fig:bench-nymgen,fig:bench-verify}. From the benchmark results we can observe that over 100 runs the runtime over $\NymGen$ averages at $4.94 \pm 0.02$ ms and $\NymVerify$ averages at $7.61 \pm 0.03$ ms.

\begin{figure}[ht]
\centering
  \includegraphics[width=0.47\textwidth,trim={0 0.9cm 0 1cm,clip}]{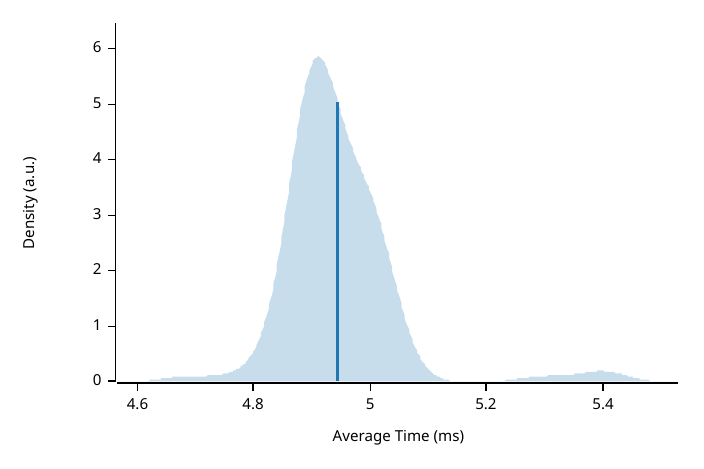}
  \caption{Benchmark results for $\NymGen$.}% of our \bpkplus\ instantiation.}
  \label{fig:bench-nymgen}
\end{figure}

\begin{figure}[ht]
\centering
  \includegraphics[width=0.47\textwidth,trim={0 0.9cm 0 1cm,clip}]{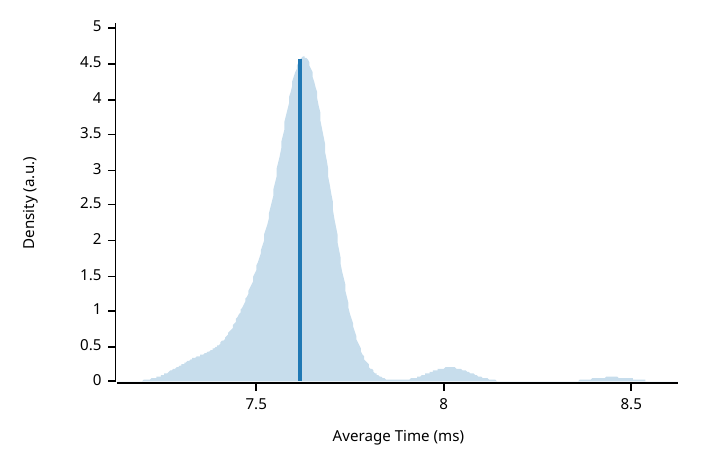}
  \caption{Benchmark results for $\NymVerify$}.% of our \bpkplus\ instantiation.}
  \label{fig:bench-verify}
\end{figure}

\paragraph{Implementation of Central Authority and Service Providers.}
Note that for the central authority it is of paramount importance that its secret key is stored in an HSM. With our scheme, all operations that need to be performed by the central authority consist of group operations in either $\GG_1$ or $\GG_2$. Furthermore, observe that the secret key of $\Groth_1$ is only applied to group elements in $\GG_1$. Hence, we can follow the same approach as in \cite{CCS:HanSla21} which use existing implementations of RSA and elliptic curve groups to implement $\GG_1$ operations on a secure element whereas all operations in $\GG_2$ can be performed outside.

For a Service Provider having access to their secret key, the situation is even simpler. The only operation that needs to be performed in an HSM in this case is one group operation to generate $\nym$. Hence, the same approach can be applied.

\section{Key Rotation and Revocation}\label{s:discussions}
  \olive{As discussed earlier, participation in the bPk system is mandatory by the Austrian eGovernment law, and thus no opt-out functionality is required and not provided.
  Also, no features for revocation or key rotation are foreseen in the current architecture, and similarly our decentralized design does not consider it either.
  However, following modern design paradigms, we next explain how to avoid impersonation attacks and how to achieve forward unlinkability in case of a key compromise by adding mechanisms for updating secret keys. In particular, we could additionally use time-dependent epoch and a stable user identifier in the user's secret key secret. A detailed description of the proposed mechanisms and their implications is provided in
 \cref{sec:key-rot}.}

\section{Conclusion}\label{s:conclusion}
In this paper, we have presented \bpkplus, a user-centric delegatable pseudonym system preserving all required key functionalities of Austria's governmental sector-specific personal identifier (bPk) concept.
We presented the first formal modeling of such a functionality, and provided a generic construction together with a specific instantiation from standard building blocks.
Our reference implementation underpins the practicability of our approach, resulting in runtimes below $20$ms for all operations on standard hardware.

\paragraph{Future Work.}
As our scheme may potentially find applications also in other national eID systems beyond Austria, future proofing the deployment with post-quantum building blocks is of central interest. Our generic construction provides the framework to instantiate the scheme with post-quantum secure building block in a black-box manner, e.g., from lattice-based signatures, public-key encryption schemes and proof systems. While recent progress was made in this area, e.g., \cite{AC:LNPT20,C:LyuNguPla22,C:ACLMT22,CCS:BLNS23}, selecting compatible proof systems and building blocks to obtain a practically efficient instantiation poses an interesting open problem.

\olive{An orthogonal interesting extension inspired, e.g., by \cite{camenisch2017privacy}, could also be to include transparency mechanisms, to allow users to track and audit all pseudonym computations that were carried out on their behalf by any of the other entities (CA or SPs) in the system.}

\smallskip

\noindent\textbf{Acknowledgements.} 
This work was in part supported by the European Union’s Horizon Europe project {\sc licorice} (grant agreement no. 101168311), {\sc prepared}, a project funded by the Austrian security research programme KIRAS of the Federal Ministry of Finance (BMF), and the CHIST-ERA project \textsc{reminder} (grant number PCI2023-145989-2) through the Austrian Science Fund (FWF) project number I 6650-N.
Views and opinions expressed are however those of the authors only and do not necessarily reflect those of the funding agencies. Neither European Union nor the granting authorities can be held responsible for them

\bibliographystyle{ACM-Reference-Format}
\bibliography{bibfile,abbrev3,crypto}

\appendix

\section{Security Proofs}\label{s:proofs}
In the following we now provide detailed formal proofs for our main security properties.
We omit a proof of correctness, as this can trivially be verified based on the construction.
\subsection{Non-frameability (\cref{thm:nonframeability})}\label{s:proof:nonframeability}
\begin{proof}
    A successful adversary $\adv$ for non-frameability against our scheme is able to produce a valid proof $\NIZK$ that opens to a honest $upk$.

    \textbf{Game 0:} This is the original non-frameability game.

    \textbf{Game 1:} As Game 0 but the generation of parameter for the ZKPoK system and the proofs are simulated in $\Oracle^{Nym}$. By reduction, we can show that an adversary $\adv$ able to distinguish this replacement can be used by an adversary $\mathcal{B}$ to break the simulation sound property of $\NIZK$. We give $\mathcal{B}$ the access to $\Oracle^{Nym}$. All other values are generated honestly. $\mathcal{B}$ gives to $\Oracle^{Nym}$ a statement and get a proof. Then he passes the proof to $\adv$. If $\adv$ notices a difference, so does $\mathcal{B}$ with the same probability. Thus, we have : $\big|\mathbb{P}(G_0) - \mathbb{P}(G_1)\big| \leq Adv^{SIM-Sound}_{\NIZK,\adv}(\lambda)$.
    %As game 0 but the ZKPoK $\NIZK$ in $\NymGen_\mathrm{user}$ is simulated. The procedure is organized as follows : We first simulate the Suppose an adversary $\adv$ get access to random oracle $\mathcal{O^\mathcal{S}}$ that output simulated valid proof. Then, by reduction, $\mathcal{O^\mathcal{S}}$ can be called by $\adv$ to generate a valid proof $\pi$ but with wrong witness. This would break the soundness of $\NIZK$. In this step, we have
    %$\big|\mathbb{P}(G_0) - \mathbb{P}(G_1)\big| \leq Adv^{Sound}_\NIZK(\lambda)$.\par

    \textbf{Game 2:} As Game 1, but we abort if we are not able to extract a witness for a valid statement. That is, we abort if we cannot extract $\usk^*,\upk^*,\sigma^*$ and $r^*$ such that the following statements are true:
        \begin{enumerate}
            \item $\Sig.\Verify(\mpk,\sigma^*,\upk^*)$ = 1        where $\sigma^*$ is part of $\usk^*$
            \item $c^* = \Enc.\Encrypt(\mpk,\upk^*;r^*)$
            \item $\upk^* = \mu(\usk'^*)$
            \item $\nym^* = \NIKE.\ShareKey(\sppk,\usk'^*)$ 
        \end{enumerate}
        By the weak simulation-sound extractability property of our proof system, we know that this can only happen with negligible probability, and thus $\big|\mathbb{P}(G_1) - \mathbb{P}(G_2)\big| \leq Adv^{Weak-Ext}_{\NIZK,\adv}(\lambda)$.

    \textbf{Game 3:}
    As the previous game, but we abort if $\upk^*$ has not been generated by the central authority, i.e., it does not belong to an honestly generated user;
    otherwise we abort.
    Now (1), we know that $\sigma^*$ is a valid signature on $\upk^*$. Therefore, the probability to abort is bounded by the probability that $\adv$ breaks the unforgeability of $\Sig$, and thus $\big|\mathbb{P}(G_2) - \mathbb{P}(G_3)\big| \leq Adv^{EUF-CMA}_{\Sig,\adv}(\lambda)$.

    \textbf{Game 4:} As Game 3, but we abort if $\adv$ creates a proof $\pi$ which point to a user public key that was never generated by the central authority. By assumption, the adversary knows $\sigma^*$. Since $\sigma^*$ is part of $\usk^*$ and cannot be forged in an efficient time, one can be deduced that the only way for the adversary to know the signature of $\upk$ is to corrupt $\usk$. $\big|\mathbb{P}(G_3) - \mathbb{P}(G_4)\big|$ = 0 follows.
    
    This proves that our scheme is non-frameable.
\end{proof}

\subsection{Anonymity (\cref{thm:anonymity})}\label{s:proof:anonymity}

\begin{proof}
We define a sequence of games to show that the advantage of an adversary $\adv$ to break the anonymity experiment is negligible.
    
\textbf{Game 0}: This is the original anonymity game.\par

\textbf{Game 1}: As Game 0 but the ZKPoP is simulated  in $\NymGen_{user}$ using $SIM = \{\mathcal{S}_1, \mathcal{S}_2\}$. First, we modify the setup by using $\mathcal{S}_1$ to generate the public parameters and a simulated trapdoor $\tau$. Then, we simulate the proof in $\NymGen_{user}$ with $\mathcal{S}_2$ and $\tau$ as in \cref{simproof}. We update $\Oracle^{LoR}$ and $\Oracle^{Nym}$ by replacing $\pi$ with a simulated proof $\pi^*$ as in \cref{Game1LorNym}. We can show by reduction that an adversary $\adv$ able to distinguish this modification in our scheme can be used by an adversary $\mathcal{B}$ to break zero-knowledge property of $\NIZK$. $\mathcal{B}$ is given access to a prove-oracle that output a proof from a statement in input. Then, it passes the statement (the statements used are always legit) to the prove-oracle. Then, $\mathcal{B}$ passes the generated proofs to $\adv$. If $\adv$ is able to detect a difference, so does $\mathcal{B}$ with the same probability. Thus, we have : $\big|\mathbb{P}(G_0) - \mathbb{P}(G_1)\big| \leq Adv^{ZK}_{\NIZK,\adv}(\lambda)$.\par

%By the zero-knowledge property of NIZK, the adversary gets no other information than the proof is true. From the adversary point of view, $\NIZK^*$ is indistinguishable from $\NIZK$.\par 

\begin{figure}[ht!]
\begin{itemize}
    \item \underline{$\ParGen^*$($1^\lambda$)}:
    \begin{algorithmic}[1]
                \State \fbox{$\pp_\NIZK,\tau \gets \mathcal{S}_1(1^\lambda)$}
                \State \Return $\pp = \{\pp_\Sig,\pp_{\Enc},\pp_\NIZK,\pp_{\mathrm{\NIKE}}\},\tau$
    \end{algorithmic}
    \item \underline{$\NymGen^*_{user}
    (\usk,\upk,\sppk,\mpk,\tau)$}:
    \begin{algorithmic}[1]
           \State $\nym \gets \NIKE.\ShareKey(\sppk,\usk')$
           \State $c \gets \Enc.\Encrypt(\mpk,\upk;r)$ with $r \getsr \mathcal{R}$
           \State \fbox{$\pi^* \gets (c,\mathrm{S}_2(\pp,\tau,x^*)$)}
           \State \Return $\nym,\pi^*$
    \end{algorithmic}
\end{itemize}
\caption{Modification of $\ParGen$ and $\NymGen_{user}$ for Game 1.}
    \label{simproof}
\end{figure}

\begin{figure}[ht!]
\begin{itemize}
    \item \underline{$\Oracle^{Nym}(\upk,\sppk)\colon$}
    \begin{algorithmic}[1]\setcounter{ALG@line}{3}
        \State \fbox{$\nym,\pi^* \gets \NymGen_{user}^*(\usk,\upk,\sppk,\mpk,\tau)$}
    \end{algorithmic}
    \item \underline{$\Oracle^{LoR}(\{\upk_0,\upk_1\},\sppk)\colon$}
    \begin{algorithmic}[1]\setcounter{ALG@line}{4}
        \State \fbox{$\nym_b,\pi^* \gets \NymGen_{user}^*(\usk_b,\upk_b,\sppk,\mpk,\tau)$}
    \end{algorithmic}
\end{itemize}
\caption{Modification of line 4 of $\Oracle^{Nym}$ and $\Oracle^{LoR}$ described in \Cref{o:nym,o:lor}.}
    \label{Game1LorNym}
\end{figure}

\textbf{Game 2}: As in Game 1, but we replace $upk$ by 0 in the encryption process inside $\NymGen_{user}$ (see \Cref{replenc}). Updates on $\Oracle^{LoR}$ and $\Oracle^{Nym}$ follows and all other values are generated as in the prior hop. Notice that the adversary does not have any access to an open-oracle since this operation is only done by the CA. Thus, It is straightforward to see that, by reduction, if an adversary $\adv$ able to distinguish this modification in our scheme can be used by an adversary $\mathcal{B}$ to break the IND-CPA security of $\Enc$.  Thus, we have : $\big|\mathbb{P}(G_1) - \mathbb{P}(G_2)\big| \leq Adv^{IND-CPA}_{\Enc,\adv}(\lambda)$.\par

%By the IND-CPA property of the encryption, an attacker cannot get any information about the message from the cipher.
\begin{figure}[ht!]
\begin{itemize}
    \item  \underline{$\NymGen^*_{user}(\usk,\upk,\sppk,\mpk)$}:
    \begin{algorithmic}[1]
           \State $\nym \gets \NIKE.\ShareKey(\sppk,\usk')$
           \State $c^* \gets \Enc.\Encrypt(\mpk,\fbox{0}; r)$
           \State $\pi^* \gets (c^*,\mathrm{S}_2(\pp_\NIZK,\tau,x^*))$
           \State \Return $\nym,\pi^*$
    \end{algorithmic}
\end{itemize}
\caption{Modification of $\NymGen_{user}$ for Game 2. }
    \label{replenc}
\end{figure}

\textbf{Game 3}: As Game 2, but in $\NymGen_{user}$, we generate $\nym$ by using a non-zero random element $\alpha \in \mathcal{K}$ instead of $\usk$ where $\mathcal{K}$ is the set in which $\usk$ is defined. For consistency, we set up a list $L$ that map $\upk$ to $\alpha$ and is updated when a user public key is used for the first time. Those changes are detailed in \Cref{Game3Nym}. We can proceed by reduction to show that if an adversary $\adv$ is able to distinguish this replacement, $\mathcal{A}$ can be used by an adversary $\mathcal{B}$ to break the indistinguishability of $\NIKE$. $\big|\mathbb{P}(G_2) - \mathbb{P}(G_3)\big| \leq Adv^{IND}_{\NIKE,\adv}(\lambda)$ follows. 

\begin{figure}[ht!]
\begin{itemize}
    \item  \underline{$\NymGen^*_{user}(\usk,\upk,\sppk,\mpk)$}:
    \begin{algorithmic}
            \If{\fbox{$\upk \notin L$}}
            \State \fbox{$\alpha \getsr \mathcal{K}^*$}
            \State \fbox{$L[\upk] \gets \alpha$}
            \Else
            \State \fbox{$\alpha \gets L[\upk]$}
            \EndIf
           \State $\nym^* \gets \NIKE.\ShareKey(\sppk,$\fbox{$\alpha$})
           \State $c^* \gets \Enc.\Encrypt(0,\mpk;r)$
           \State $\pi^* \gets (c^*,\mathrm{S}_2(\pp_\NIZK,\tau,x^*)$)
           \State \Return $\nym^*,\pi^*$
        \end{algorithmic}
    \end{itemize}
\caption{Modification of $\NymGen_{user}$ for Game 3.}
    \label{Game3Nym}
\end{figure}
The outputs of $\Oracle^{LoR}$ are now independent of the challenge bit $b$. This proves that our scheme achieves anonymity.
\end{proof}

\section{Sequence Diagrams}\label{sec:sequence-diagrams}
\olive{In \cref{fig:bpks-sequence-diagrams}, we provide an overview of the logical flows of the current and envisioned phases in the Austrian bPk system.}

\olive{Note that in both systems, the user key is generated when the user first registers in Austria; however, in the existing system, this key is never returned to the user.
In a new \emph{key request} step, citizens may obtain their keys after proper authentication using, e.g., the national eID system.
In the standard authentication flow, the central authority is no longer involved in our new design; yet, while not depicted here, the legacy flow is functionally still supported, e.g., for less tech-savvy users which prefer a ``bPk-as-a-service'' setting.
In both versions, service providers can further request the central authority to compute pseudonyms for specific users after providing proof of eligibility to obtain these pseudonyms.
While this is a common flow in the current system (in particular by public agencies), this can now also be done locally by dedicated (i.e., semi-trusted) service providers, which can compute pseudonyms for users \emph{for their own scope}, thereby further reducing the dependency on the central authority.}

\begin{figure*}
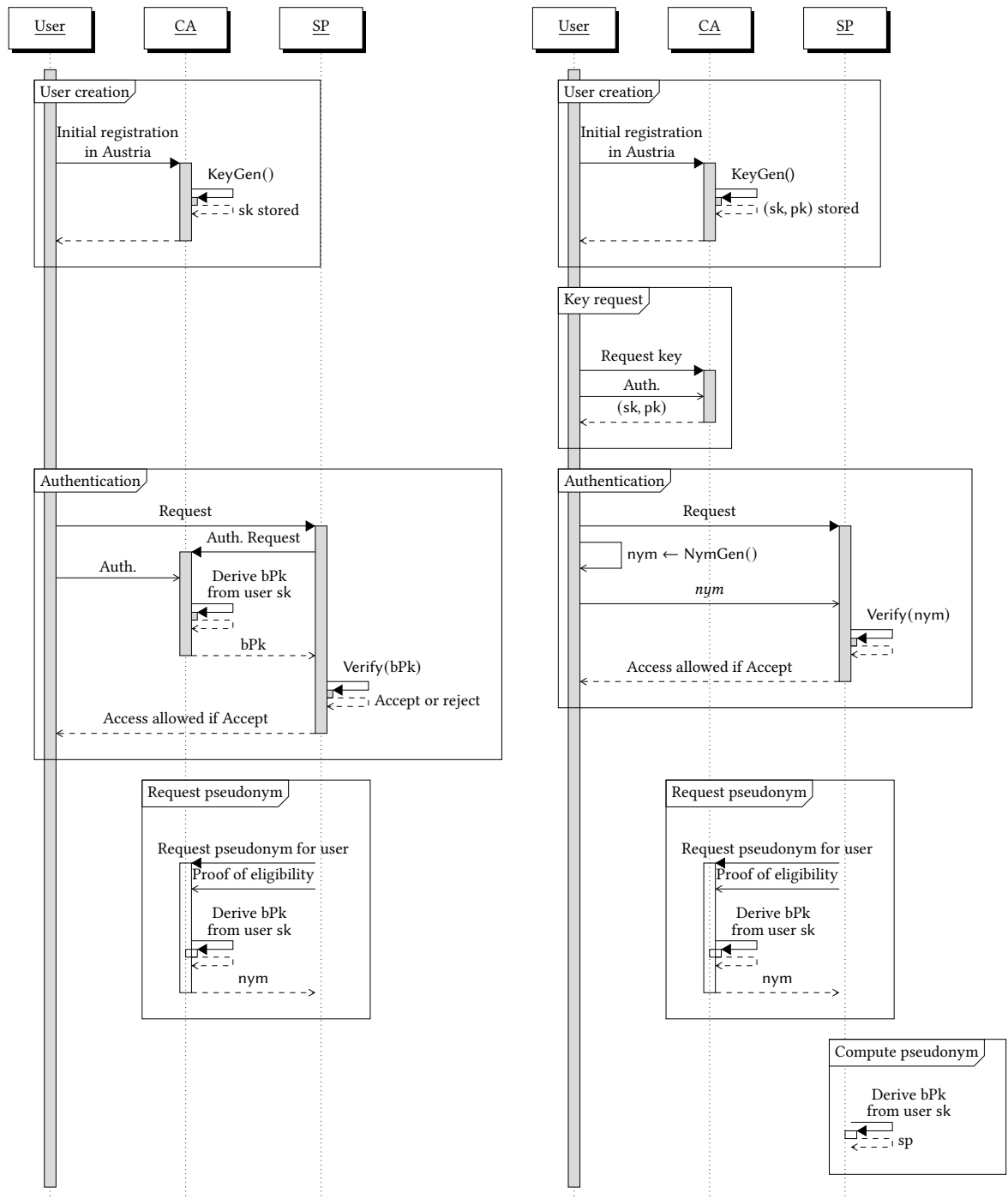
%[h]
\begin{minipage}[t]{0.45\textwidth}
    \scalebox{0.80}{
    \begin{sequencediagram}
        \renewcommand\unitfactor{0.5}
        \newthread{u}{User}
        \newinst[1]{ca}{CA}
        \newinst[1]{sp}{SP}

    \begin{sdblock}{User creation}{}
        
    \stepcounter{seqlevel}
    \begin{call}{u}{\shortstack{Initial registration\\in Austria}}{ca}{}
            
            \begin{callself}{ca}{$\KeyGen()$}{$\sk$ stored}
            \end{callself}
        \end{call}    
    \end{sdblock}

          \stepcounter{seqlevel}
          \stepcounter{seqlevel}
          \stepcounter{seqlevel}
          \stepcounter{seqlevel}
          \stepcounter{seqlevel}
          \stepcounter{seqlevel}
          \stepcounter{seqlevel}

    \begin{sdblock}{Authentication}{}
    \begin{call}{u}{Request}{sp}{Access allowed if Accept}
        \begin{call}{sp}{Auth. Request}{ca}{\bpk}
            
            \mess{u}{Auth.}{ca}
            \begin{callself}{ca}{\shortstack{Derive \bpk\\ from user $\sk$}}{}
            \end{callself}
            
        \end{call}
       
        %\mess{sp}{Auth.}{sp}
        \begin{callself}{sp}{$\Verify(\mathrm{\bpk})$}{Accept or reject}
        \end{callself}
    \end{call}
    \end{sdblock}

    \begin{sdblock}{Request pseudonym}{}
          \stepcounter{seqlevel}
            \begin{call}{sp}{\shortstack{Request pseudonym for user}}{ca}{$\nym$}
            \mess{sp}{Proof of eligibility}{ca}
            \stepcounter{seqlevel}
            \begin{callself}{ca}{\shortstack{Derive \bpk\\ from user $\sk$}}{}
            \end{callself}
        \end{call}    
    \end{sdblock}
    
           \stepcounter{seqlevel}
          \stepcounter{seqlevel}
          \stepcounter{seqlevel}
          \stepcounter{seqlevel}
             \stepcounter{seqlevel}
          \stepcounter{seqlevel}

   \end{sequencediagram}
    }
    \end{minipage}
    \begin{minipage}[t]{0.45\textwidth}
    \scalebox{0.80}{       
    \begin{sequencediagram}
        \renewcommand\unitfactor{0.5}
        \newthread{u}{User}
        \newinst[1]{ca}{CA}
        \newinst[1]{sp}{SP}    
        \begin{sdblock}{User creation}{}
    \stepcounter{seqlevel}
    \begin{call}{u}{\shortstack{Initial registration\\in Austria}}{ca}{}
                %\mess{ca}{Auth.}{ca}
                \begin{callself}{ca}{KeyGen()}{$(\sk,\pk)$ stored}
                \end{callself}
            \end{call}    
        \end{sdblock}

        \begin{sdblock}{Key request}{}
          \stepcounter{seqlevel}
            \begin{call}{u}{Request key}{ca}{$(\sk,\pk)$}
                \mess{u}{Auth.}{ca}
            \end{call}
        \end{sdblock}
        
        \begin{sdblock}{Authentication}{}
            \begin{call}{u}{Request}{sp}{Access allowed if Accept}
            \mess{u}{$\nym\gets\NymGen()$}{u}
            \mess{u}{$nym$}{sp}
            \begin{callself}{sp}{$\Verify(\nym)$}{}\end{callself}
            \end{call}
        \end{sdblock}

          \stepcounter{seqlevel}
          \stepcounter{seqlevel}

        \begin{sdblock}{Request pseudonym}{}
          \stepcounter{seqlevel}
            \begin{call}{sp}{\shortstack{Request pseudonym for user}}{ca}{$\nym$}
            \mess{sp}{Proof of eligibility}{ca}
            \stepcounter{seqlevel}
            \begin{callself}{ca}{\shortstack{Derive \bpk\\ from user $\sk$}}{}
            \end{callself}
        \end{call}    
    \end{sdblock}

    \begin{sdblock}{Compute pseudonym}{}
          \stepcounter{seqlevel}
            \begin{callself}{sp}{\shortstack{Derive \bpk\\ from user $\sk$}}{sp}
            \end{callself}
    \end{sdblock}
    
\end{sequencediagram}
    }
    \end{minipage} 
    \caption{\olive{
High-level sequence diagrams for the existing (left) and proposed (right) bPk system.}}
    \label{fig:bpks-sequence-diagrams}
\end{figure*}

\section{Rotation key and Revocation}\label{sec:key-rot}

\olive{
  The signature in the user's secret key, i.e., $\sigma \gets \Sig.\Sign(\upk,\msk_{\Sig})$, could be extended to $\sigma \gets \Sig.\Sign((\upk,\mathtt{uid},\mathtt{epoch}),\msk_{\Sig})$, where $\mathtt{uid}$ is a stable user identifier attached to a user's public key, and $\mathtt{epoch}$ is a counter increasing on a regular basis.
  Upon authentication, users would now prove knowledge of a valid signature \emph{for the current $\mathtt{epoch}$} without disclosing $\mathtt{uid}$.
  In case that linking to previous periods is necessary, the user could additionally send the pseudonym from that previous period and show that they are both belonging to the same $\texttt{uid}$ while keeping it private.
  Adapting the zero-knowledge proofs in our construction is straightforward at low costs. The computational costs will roughly double if a user needs to prove consistency with a previous pseudonym which only needs to done once per service.} 

  \olive{When a user registers to a new service and needs to prove that it is the first registration, we would recommend to leverage the central authority to generate all previous pseudonyms of that user so that the service provider can check them against its existing database.
  While coming at the cost of a somewhat more active central authority, this approach relieves the user from having to store all their previous secret keys.}

  \olive{By ensuring that user secret keys are only stored within secure hardware elements and appropriate authentication mechanisms are employed, e.g., by requiring biometric authentication to request access -- the latter being a standard requirement in many eID solutions -- this approach enables key rotation at relatively low costs.
  However, if ``on the spot'' revocation of a user key is required, more sophisticated solutions need to be applied, leveraging from the broad literature on (issuer-side) revocation of attribute-based credential systems using very different approaches including, e.g., accumulators~\cite{PKC:CamKohSor09,C:CamLys02,EUROSP:BCDLRSY17} or allow/deny lists~\cite{PKC:NFHF09}.
  There, the issuer typically maintains some public revocation information on all active (alternatively: revoked) credentials, and users upon authentication prove that their credential is currently valid (or not revoked) relative to this revocation information, without revealing any further information about their certificate.
  However, revocation of anonymous credentials often introduces substantial computational and procedural overheads, such that a clear risk assessment under the current legislation would be required.}

\end{document}